\documentclass[11pt,draft]{article}

\makeatletter
\renewcommand*\@fnsymbol[1]{\the#1}
\makeatother

\usepackage{amsmath}
\usepackage{amsfonts}
\usepackage{amsthm}
\usepackage[english]{babel}
\usepackage{latexsym}
\usepackage[hmargin=2cm,vmargin=4cm]{geometry}
\usepackage{enumerate}
\usepackage{nicefrac}
\usepackage{mathrsfs}
\usepackage[affil-it]{authblk}

\theoremstyle{plain}
\newtheorem{theorem}{Theorem}[section]
\newtheorem{lemma}[theorem]{Lemma}
\newtheorem{proposition}[theorem]{Proposition}
\newtheorem{corollary}[theorem]{Corollary}
\theoremstyle{definition}
\newtheorem{definition}[theorem]{Definition}
\theoremstyle{remark}
\newtheorem{remark}[theorem]{Remark}
\newtheorem{example}[theorem]{Example}

\newcommand{\Dom}{\mathop{{\rm dom}}\nolimits}
\newcommand{\Epi}{\mathop{\rm epi}\nolimits}
\newcommand{\Interior}{\mathop{\rm int}\nolimits}
\newcommand{\Core}{\mathop{\rm core}\nolimits}
\newcommand{\e}{\varepsilon}
\newcommand{\VaR}{\mathop {\rm VaR}\nolimits}

\newcommand{\TVaR}{\mathop {\rm TVaR}\nolimits}

\newcommand{\ccI}{\mathcal I}
\newcommand{\cF}{\mathcal F}

\newcommand{\probp}{\mathbb P}

\newcommand{\cA}{\mathscr A}

\newcommand{\cS}{\mathscr S}
\newcommand{\cX}{\mathscr X}
\newcommand{\cY}{\mathscr Y}
\newcommand{\cB}{\mathscr B}

\newcommand{\cU}{\mathscr U}
\newcommand{\R}{\mathbb R}
\newcommand{\N}{\mathbb N}
\newcommand{\E}{\mathbb E}

\begin{document}

\title{
Beyond cash-additive risk measures: when changing the num\'{e}raire fails
}

\author{\sc{Walter Farkas}\thanks{
Email: \texttt{walter.farkas@bf.uzh.ch}}
\,, 
\sc{Pablo Koch-Medina}\thanks{
Email: \texttt{pablo.koch@bf.uzh.ch}}}
\affil{Department of Banking and Finance, University of Zurich, Switzerland}

\author{\sc{Cosimo Munari}\,\thanks{
Email: \texttt{cosimo.munari@math.ethz.ch}}
}
\affil{Department of Mathematics, ETH Zurich, Switzerland}

\date{September 12, 2013\\
 \ \\
To appear in \sc{Finance and Stochastics}}

\maketitle

\abstract{
We discuss risk measures representing the minimum amount of capital a financial
institution needs to raise and invest in a pre-specified \textit{eligible asset} to
ensure it is adequately capitalized. Most of the literature has focused on
cash-additive risk measures, for which the eligible asset is a risk-free bond, on the
grounds that the general case can be reduced to the cash-additive case by a change of
num\'{e}raire. However, discounting does not work in all financially relevant
situations, typically when the eligible asset is a defaultable bond. In this paper we
fill this gap allowing for general eligible assets. We provide a variety of
finiteness and continuity results for the corresponding risk measures and apply them
to risk measures based on Value-at-Risk and Tail Value-at-Risk on $L^p$ spaces, as
well as to shortfall risk measures on Orlicz spaces. We pay special attention to the
property of cash subadditivity, which has been recently proposed as an alternative to
cash additivity to deal with defaultable bonds. For important examples, we provide
characterizations of cash subadditivity and show that, when the eligible asset is a
defaultable bond, cash subadditivity is the exception rather than the rule. Finally,
we consider the situation where the eligible asset is not liquidly traded and the
pricing rule is no longer linear. We establish when the resulting risk measures are
quasiconvex and show that cash subadditivity is only compatible with continuous
pricing rules.
}

\bigskip


\medskip


\medskip




\parindent 0em \noindent

\section{Introduction}
\label{intro}

\subsubsection*{Motivation}

Risk measures in their current axiomatic form were essentially introduced by Artzner,
Delbaen, Eber and Heath in their landmark paper~\cite{ArtznerDelbaenEberHeath1999}.
In that paper, the authors consider a one-period economy with dates~$t=0$ and~$t=T$
where future financial positions, or net worths, are represented by elements of the
space~$\cX$ of random variables on a finite measurable space. A financial institution
with future net worth $X\in\cX$ is considered to be adequately capitalized if~$X$
belongs to a pre-specified set $\cA\subset\cX$ satisfying the axioms of a
\textit{(coherent) acceptance set}. Once a \textit{reference asset} $S=(S_0,S_T)$
with initial price $S_0>0$ and positive terminal payoff $S_T\in\cX$ has been
specified, the corresponding \textit{risk measure} is defined by setting
\begin{equation*}
\label{eq intro}
\rho_{\cA,S}(X):=\inf\left\{m\in\R \,; \ X+\frac{m}{S_0} S_T\in\cA\right\}\,.
\end{equation*}

As articulated in~\cite{ArtznerDelbaenEberHeath1999}, the idea behind risk measures
is that ``sets of acceptable future net worths are the primitive objects to be
considered in order to describe acceptance or rejection of a risk. [...]
\textit{given} some `reference instrument', there is a natural way to define a
measure of risk by describing how close or how far from acceptance a position is''.

\smallskip

In terms of capital adequacy the interpretation is that, whenever finite and
positive, the number $\rho_{\cA,S}(X)$ represents the minimum amount of capital the
institution needs to raise and invest in the reference asset to become adequately
capitalized. If finite and negative, then $-\rho_{\cA,S}(X)$ represents the maximum
amount of capital the institution can return without compromising its capital
adequacy.

\smallskip

The theory of coherent risk measures was extended to general probability spaces in
Delbaen~\cite{Delbaen2000}. In that paper, the focus is on \textit{cash-additive}
risk measures, i.e. risk measures where the reference asset is the risk-free asset
$S=(1,1_\Omega)$ with risk-free rate set to zero. Hence, the risk measure is given
by
\begin{equation*}
\label{eq intro 0}
\rho_{\cA}(X):=\rho_{\cA,S}(X)=\inf\left\{m\in\R \,; \ X+m1_\Omega\in\cA\right\}\,.
\end{equation*}

In the remark after Definition~2.1, Delbaen refers
to~\cite{ArtznerDelbaenEberHeath1999} for an interpretation and notes that ``here we
are working in a model without interest rate, the general case can `easily' be
reduced to this case by `discounting'\,''.

\smallskip

The theory of risk measurement has since then been extended in many directions and,
not surprisingly, based on the above discounting argument, most of the literature has
focused on cash-additive risk measures. Yet, this exclusive focus on cash additivity
is only justified if \textit{every} economically meaningful situation can be reduced
to the cash-additive setting. This, however, is by no means the case.

\smallskip

To see this it is useful to make the discounting argument explicit. Consider an
infinite probability space $(\Omega,\cF,\probp)$ and assume the space of future
financial positions is $\cX=L^p$ for some $0\leq p\leq\infty$. Take an acceptance set
$\cA\subset L^p$ and a reference asset $S=(S_0,S_T)$, where $S_T\in L^p$ is a
nonzero, positive terminal payoff. If~$S_T$ is (essentially) bounded away from zero,
i.e. $S_T\geq\e$ almost surely for some $\e>0$, we can use~$S$ as the new
num\'{e}raire and consider ``discounted'' positions $\widetilde{X}:=X/S_T$. Note
that, in this case, discounted positions still belong to~$\cX$. Setting
$\widetilde{\cA}:=\{X/S_T \,; \ X\in\cA\}$, it is easy to see that
\begin{equation*}
\rho_{\cA,S}(X)=S_0\,\rho_{\widetilde{\cA}}\,(\widetilde{X})\,.
\end{equation*}
Hence, in this case the risk measure~$\rho_{\cA,S}$ can be reduced to a cash-additive
risk measure acting on ``discounted'' positions. However, this reduction fails
whenever the payoff of the reference asset is not bounded away from zero.

\begin{enumerate}
	\item If $\probp(S_T=0)>0$, then~$S$ does not qualify as a num\'{e}raire and
the ``discounting'' procedure is not applicable.

	\item If $\probp(S_T=0)=0$ but~$S_T$ is not bounded away from zero, then we can
use~$S$ as a num\'{e}raire but ``discounted'' positions will typically no longer
belong to~$L^p$, unless~$p=0$. Moreover, any choice of the space of discounted
positions will depend on the particular choice of the num\'{e}raire asset.
\end{enumerate}

Reference assets whose payoffs are not bounded away from zero arise in situations
which are not uncommon in financial applications. For instance, the payoff of
\textit{shares} is typically modeled by random variables which are not bounded away
from zero, such as random variables with lognormal or L\'{e}vy distribution. Perhaps
more importantly, the same is true of \textit{defaultable bonds}. Indeed, assume that $S=(S_0,S_T)$ is a defaultable bond with face value~$1$ and price $S_0<1$. The
payoff~$S_T$ corresponds to a random variable taking values in the interval $[0,1]$
and can be interpreted as the \textit{recovery rate}. Depending on what the recovery
rate is in the various states of the economy,~$S_T$ can be bounded away from zero or
not and even assume the value zero in some future scenario. In particular, the case
of zero recovery might describe the situation when actual recovery is positive but
occurs only \textit{after} time~$t=T$.

\smallskip

Bearing in mind the above mentioned interpretation given
in~\cite{ArtznerDelbaenEberHeath1999}, it is clear that acceptability is the key
concept and that when measuring the distance to acceptability we should not restrict
\textit{a priori} the range of possible reference assets. Therefore, it is important
to go beyond cash-additive risk measures and to investigate risk measures with
respect to a general reference asset whose payoff is not necessarily bounded away
from zero. Moreover, we consider acceptance sets that are not necessarily coherent or
convex. This allows us to cover, for example, risk measures based on Value-at-Risk
acceptability, which is the most widely used acceptability criterion in practice.

\subsubsection*{Setting and main results}

In this paper the space~$\cX$ of financial positions at time~$t=T$ is assumed to be a
general ordered topological vector space with positive cone~$\cX_+$. The set of
acceptable future positions~$\cA$ is taken to be any nontrivial subset of~$\cX$
satisfying $\cA+\cX_+\subset\cA$, and the reference asset $S=(S_0,S_T)$ is described
by its unit price $S_0>0$ and its nonzero terminal payoff $S_T\in\cX_+$. This setup
is general enough to cover the whole range of spaces commonly encountered in the
literature, and to incorporate all financially relevant situations that cannot be
captured within the standard cash-additive framework.

\smallskip

A comment on our choice to work on general topological vector spaces is in order.
Since the typical spaces used in applications -- $L^p$ and Orlicz spaces -- are
Fr\'{e}chet lattices, one might argue that it is sufficient to restrict the attention
to this type of spaces. The motivation for a more abstract setting is twofold. First,
a genuine mathematical interest in understanding the minimal structure required to
support a theory of risk measures. Second, even when working within a  Fr\'{e}chet
lattice setting, one is sometimes led to equip the underlying space with a different
topology -- for instance, to obtain the special dual representations in Biagini and
Frittelli~\cite{BiaginiFrittelli2009} or in Orihuela and Ruiz
Galan~\cite{OrihuelaRuiz2012} or, in particular, to deal with risk measures
on~$L^\infty$ having the Fatou property, which is nothing but lower semicontinuity
with respect to the $\sigma(L^\infty,L^1)$ topology. This immediately takes us
outside the domain of Fr\'{e}chet lattices.

\smallskip

In this general context, we will address the following issues.

\smallskip

\textit{Finiteness}. Given our interpretation of risk measures as required capital,
it is important to study finiteness properties. Indeed, if $\rho_{\cA,S}(X)=\infty$
for a position $X\in\cX$, then~$X$ cannot be made acceptable by raising any amount of
capital and investing it in the reference asset~$S$. This means that~$S$ is not an
effective vehicle to help reach acceptability for that position. On the other hand,
if $\rho_{\cA,S}(X)=-\infty$ then we could extract arbitrary amounts of capital
without compromising the acceptability of~$X$, which is financially implausible. Note
also that in many cases it is possible to establish that finiteness implies
continuity -- as for convex risk measures on Fr\'{e}chet lattices due to the extended
Namioka-Klee theorem in Biagini and Frittelli~\cite{BiaginiFrittelli2009} or, in a
more general setting, in Borwein~\cite{Borwein1987}. Thus, understanding finiteness
is also relevant from this perspective. Note that, since no finiteness result is
provided in~\cite{BiaginiFrittelli2009}, our finiteness results can be considered to
be complementary to that paper.

\smallskip

\textit{Continuity and dual representations}. Much effort in the literature has been
devoted to showing various continuity properties of risk measures. From a practical
perspective continuity is important since if~$\rho_{\cA,S}$ fails to be continuous at
some position~$X$, then a slight change or misstatement of~$X$ might lead to a
dramatically different capital requirement. Moreover, as recently discussed in
Kr\"{a}tschmer, Schied and Z\"{a}hle~\cite{KreatschmerSchiedZaehle2013}, continuity
is closely related to statistical robustness. Finally, continuity is also a useful
property in the context of dual representations, which play an important role in
optimization problems, for instance arising in connection to portfolio selection.

\smallskip

We undertake a systematic investigation of finiteness and continuity in terms of the
interplay between the two fundamental financial primitives: the acceptance set~$\cA$
and the reference asset $S=(S_0,S_T)$. Since we do not restrict their range a priori,
the results in this paper are entirely new in this generality and sometimes provide
new insights even for the standard cash-additive case. The main results are the
following:
\begin{enumerate}
\item In Proposition~\ref{pseudo order unit} we provide a complete picture of
    finiteness and continuity when~$S_T$ belongs to the core or the interior
    of~$\cX_+$, without any assumption on~$\cA$.

\item In Theorem~\ref{finiteness in frechet lattice} we establish a sufficient
    condition for finiteness in case~$\cX$ is a topological Riesz space and~$\cA$
    has nonempty interior, extending to the non-convex case the finiteness result
    obtained in Theorem~2.3 in Svindland~\cite{Svindland2008} and in Theorem~4.6 in
    Cheridito and Li~\cite{CheriditoLi2009} for convex, cash-additive risk measures
    on $L^p$ spaces and Orlicz hearts, respectively.

\item In Theorem~\ref{cont after finit} we prove a characterization of continuity
    for convex risk measures, which can be seen as a generalization to arbitrary
    ordered topological vector spaces of the extended Namioka-Klee theorem in
    Biagini and Frittelli~\cite{BiaginiFrittelli2009} when applied to risk
    measures.

\item In Theorem~\ref{general finiteness convex with str pos elements} and
    Corollary~\ref{finiteness for convex acc} we provide criteria for finiteness
    and continuity in case~$\cA$ is convex.

\item In Theorem~\ref{finite coherent} we provide a full characterization of
    finiteness and continuity when~$\cA$ is coherent.
\end{enumerate}

\smallskip

\textit{Applications}. Throughout Section~\ref{section applications} we provide
several concrete examples. In particular, we focus on risk measures based on the most
prominent acceptability criteria in practice: acceptability based on Value-at-Risk,
on Tail Value-at-Risk, and on shortfall risk arising in the context of utility
maximization problems.

\smallskip

\textit{Cash subadditivity}.  Cash-subadditive risk measures were introduced in El
Karoui and Ravanelli~\cite{ElKarouiRavanelli2009} with the intent to ``model
stochastic and/or ambiguous interest rates or defaultable contingent claims''. Since
our framework provides a natural approach to deal with defaultable reference assets,
we investigate in Section~\ref{cash-sub section} when~$\rho_{\cA,S}$ is cash
subadditive on~$L^p$. When $S=(S_0,S_T)$ is a defaultable bond, we always have cash
subadditivity if~$S$ can only default on the interest payment, i.e. if
$\probp(S_T<S_0)=0$. For important choices of the acceptance set, we show
that~$\rho_{\cA,S}$ fails to be cash subadditive unless the probability
$\probp(S_T<S_0)$ that the invested capital is at risk is sufficiently small or
sometimes even zero. Hence, if~$\rho_{\cA,S}$ is to be cash subadditive, the bond~$S$
can only be allowed to default to a fairly limited extent. These findings provide a
better insight into the property of cash subadditivity and show that the link between
cash subadditivity and defaultability is less straightforward than what was suggested
in~\cite{ElKarouiRavanelli2009}.

\smallskip

\textit{Illiquid markets}. In Section~\ref{illiquidity section} we allow for the
possibility that the market for the reference asset is not liquid. In this case we
are naturally led to a quasiconvex risk measure, for which we provide in
Proposition~\ref{corollary on dual illiquid} a dual representation highlighting the
underlying financial fundamentals. We also show in Proposition~\ref{continuity of
pricing functional} that the associated risk measure can only be cash subadditive if
the pricing rule for the reference asset depends continuously on the traded volume.

\subsubsection*{Embedding in the literature}

Risk measures with respect to a general reference asset have been considered before
to various degrees. In addition to the seminal
paper~\cite{ArtznerDelbaenEberHeath1999} by Artzner, Delbaen, Eber and Heath, we
refer to Jaschke and K\"{u}chler~\cite{JaschkeKuchler2001}, Frittelli and
Scandolo~\cite{FrittelliScandolo2006}, Hamel~\cite{Hamel2006}, and Filipovi\'{c} and
Kupper~\cite{FilipovicKupper2007}. More recent relevant publications are the
papers~\cite{ArtznerDelbaenKoch2009} by Artzner, Delbaen, and Koch-Medina,
and~\cite{KonstantinidesKountzakis2011} by Konstantinides and Kountzakis. Some of
these references contain results on finiteness and continuity, as well as dual
representations. However, all relevant results are obtained, implicitly or
explicitly, under the assumption that the payoff of the reference asset is an
interior point of the positive cone. This critically limits their applicability since
the positive cone of many spaces encountered in the literature has empty interior --
for instance~$L^p$ spaces, $0\leq p<\infty$, and Orlicz hearts on nonatomic
probability spaces. In this respect, Proposition~\ref{pseudo order unit} can be seen
as a general formulation of that type of results. In~\cite{FarkasKochMunari2012}, the
present authors consider general eligible assets in the $L^\infty$ setting. However,
the treatment there relies heavily on the fact that the positive cone in $L^\infty$
has nonempty interior and cannot be adapted to more general spaces which are
important in financial applications. Finally, we mention that risk measures of the
form~$\rho_{\cA,S}$ on~$L^p$ can be regarded as scalarizations of set-valued risk
measures -- as studied in Hamel, Heyde and Rudloff~\cite{HamelHeydeRudloff2011} --
where the underlying market consists of~$S$ and the risk-free asset. Hence, our
results can also be applied in that particular setting.


\section{Risk measures beyond cash additivity}
\label{sec 2}

We start by defining risk measures associated to general acceptance sets and general
reference assets, setting the scene for the remainder of the paper.


\subsection{The space of financial positions}

In this paper, financial positions are assumed to belong to a (Hausdorff) topological
vector space over~$\R$ denoted by~$\cX$. We assume~$\cX$ is ordered by a pointed
convex cone~$\cX_+$ called the {\em positive cone}. Note that a set $\cA\subset\cX$
is a cone if $\lambda\cA\subset\cA$ for all $\lambda\ge 0$ and is pointed if
$\cA\cap(-\cA) =\{0\}$. We write $X\le Y$ whenever $Y-X\in\cX_+$. The topological
dual of~$\cX$ is denoted by~$\cX'$. The space~$\cX'$ is itself an ordered vector
space when equipped with the positive cone~$\cX'_+$ consisting of all functionals
$\psi\in\cX'$ such that $\psi(X)\geq 0$ whenever $X\in\cX_+$.

\smallskip

If~$\cA$ is a subset of~$\cX$, we denote by~$\Interior(\cA)$,~$\overline{\cA}$,
and~$\partial\cA$ the interior, the closure, and the boundary of~$\cA$, respectively.
Moreover, we denote by~$\Core(\cA)$ the core, or algebraic interior, of~$\cA$, i.e.
the set of all positions $X\in\cA$ such that for each $Y\in\cX$ there exists $\e>0$
with $X+\lambda Y\in\cA$ whenever $\left|\lambda\right|<\e$.

\smallskip

In case~$\cX$ is equipped with a lattice structure, we use the standard notation
$X\vee Y:=\sup\{X,Y\}$ and $X\wedge Y:=\inf\{X,Y\}$. Moreover, we set $X^+:=X\vee0$
for the positive part of~$X$, $X^-:=(-X)\vee0$ for its negative part, and
$\left|X\right|:=X\vee(-X)$ for its absolute value.

\medskip

\begin{example}[Standard spaces]
Standard examples of ordered topological vector spaces used in financial mathematics
are provided by spaces of random variables defined on a probability space
$(\Omega,\cF,\probp)$, which throughout this paper will always be assumed to be
\textit{nonatomic}. As usual, random variables which coincide almost surely are
identified so that equalities and inequalities involving random variables will always
be understood in the almost-sure sense. The natural order structure is given by
almost-sure pointwise ordering. The vector space~$L^0$ of all $\cF$-measurable
functions $X:\Omega\to\R$ is a Fr\'{e}chet lattice with respect to the topology of
convergence in probability. If $0<p<\infty$, we denote by~$L^p$ the subspace of~$L^0$
consisting of all functions satisfying $\E[\left|X\right|^p]<\infty$. It is a Banach
lattice under the usual norm when $p\geq 1$ and a Fr\'{e}chet lattice under the usual
metric when $0<p<1$. The space~$L^\infty$ is the subspace of~$L^0$ consisting of all
essentially bounded functions. It is a Banach lattice with respect to the standard
(essential) supremum norm. If~$\Phi$ is an Orlicz function as defined
in~\cite{EdgarSucheston1992}, the Orlicz space~$L^\Phi$ is the subspace of~$L^0$
consisting of all functions $X\in L^0$ such that $\E[\Phi(\lambda X)]<\infty$ for
some $\lambda>0$. The Orlicz heart~$H^\Phi$ is the subspace of~$L^\Phi$ consisting of
all functions satisfying the previous inequality for every $\lambda>0$. These spaces
are Banach lattices under the Luxemburg norm.

\smallskip

If~$\cX$ is any of the spaces described above and~$\cY$ is a vector space such that
$(\cX,\cY)$ is a dual pair, then~$\cX$ equipped with the weak topology
$\sigma(\cX,\cY)$ is an ordered topological vector space. However, it is no longer a
Fr\'{e}chet lattice by Corollary~9.9 in~\cite{AliprantisBorder2006}. A typical
instance of this situation encountered in financial mathematics is when $L^\infty$ is
equipped with the $\sigma(L^\infty,L^1)$ topology.
\end{example}

\medskip

The positive cone~$\cX_+$ may have empty interior. In this case we will consider two
types of substitutes for interior points: order units and strictly positive elements.
The elements in $\Core(\cX_+)$ are called \textit{order units}. A point $X\in\cX_+$
is called \textit{strictly positive} whenever $\psi(X)>0$ for all nonzero
$\psi\in\cX'_+$. The set of all strictly positive points will be denoted
by~$\cX_{++}$. We always have $\Interior(\cX_+)\subset\Core(\cX_+)\subset\cX_{++}$.
These inclusions are in general strict, but they coincide whenever~$\cX_+$ has
nonempty interior.

\medskip

\begin{example}
\label{ex bounded away from zero}
\begin{enumerate}[(i)]
	\item (nonempty interior) The positive cone of~$L^\infty$ has nonempty interior
and we have $X\in\Interior(L^\infty_+)$ if and only if $X\geq\e$ almost surely for some
$\e>0$. In particular one should not confuse strictly positive elements with
functions that are strictly positive almost surely.
  \item (empty interior, nonempty core) If we endow~$L^\infty$ with the weak
      topology~$\sigma(L^\infty,L^1)$, then it is not difficult to see that the
      interior of the positive cone is empty. Note that, as in~(i), any positive
      element in~$L^\infty$ which is bounded away from zero is an order unit.
      Moreover, the strictly positive elements are precisely those $X\in L^\infty$
      such that $X>0$ almost surely. As a result, the inclusion
      $\Core(L^\infty_+)\subset L^\infty_{++}$ is strict, even if the positive cone
      has nonempty core.
  \item (empty core, but strictly positive elements) The positive cone of~$L^p$,
      for $1\leq p<\infty$, has empty core. However, the elements $X\in L^p$ such
      that $X>0$ almost surely correspond to the strictly positive elements. The
      same is true for any (nontrivial) Orlicz heart~$H^\Phi$.
  \item (no strictly positive elements) It is known that strictly positive elements
      may not exist, see Exercise~10 in Section~2.2 of~\cite{AliprantisTourky2007},
      where the space~$\cX$ is a nonstandard function space. In Remark~\ref{no
      strictly positive orlicz} below, we provide a more interesting example. We
      show that the Orlicz space~$L^\Phi$ defined by
      $\Phi(x):=e^{\left|x\right|}-1$ has no strictly positive elements.
\end{enumerate}
\end{example}


\subsection{From unacceptable to acceptable}

In this section we introduce risk measures with respect to general reference assets
and general acceptance sets and establish some of their basic properties. A detailed
motivation for studying this type of risk measures was provided in the introduction.
In Section~\ref{section applications} we will discuss several examples of acceptance
sets which are relevant for financial applications.

\medskip

\begin{definition}
\label{acceptance set definition}
A set $\cA\subset\cX$ is called an {\em acceptance set} whenever the following two
conditions are satisfied:
\begin{enumerate}[(i)]
  \item $\cA$ is a nonempty, proper subset of $\cX$ (non-triviality);
  \item if $X\in\cA$ and $Y\geq X$ then $Y\in\cA$ (monotonicity).
\end{enumerate}
\end{definition}

\medskip

These conditions seem to be minimal in the sense that non-triviality allows to
discriminate between ``good'' and ``bad'' positions and monotonicity captures the
intuition that a financial institution is better capitalized than another if the net
worth of the first dominates the net worth of the second. Special classes of
acceptance sets that will be considered later are \textit{convex} acceptance sets,
\textit{conic} acceptance sets, and \textit{coherent} acceptance sets, i.e.
acceptance sets which are convex cones. We refer
to~\cite{ArtznerDelbaenEberHeath1999} and~\cite{FoellmerSchied2011} for a financial
interpretation of these special acceptance sets.

\medskip

We now consider traded assets $S=(S_0,S_T)$ with initial price $S_0>0$ and nonzero
terminal payoff $S_T\in\cX_+$. If a position $X\in\cX$ is not acceptable with respect
to a given acceptance set $\cA\subset\cX$, it is natural to ask which actions can
turn it into an acceptable position, and at which cost. In line with the definition
of a risk measure proposed in~\cite{ArtznerDelbaenEberHeath1999}, we allow for one
specific action: raising capital and investing it in a pre-specified traded asset
$S$. In the sequel, we adopt the standard notation
$\overline{\R}:=\R\cup\{\infty,-\infty\}$.

\medskip

\begin{definition}
\label{capital requirement def}
Let $\cA\subset\cX$ be a monotone set and $S=(S_0,S_T)$ a traded asset. The
\textit{risk measure} with respect to~$\cA$ and~$S$ is the function
$\rho_{\cA,S}:\cX\to\overline{\R}$ defined by
\begin{equation}
\label{eq def cap req}
\rho_{\cA,S}(X):=\inf\left\{m\in\R \,; \ X+\frac{m}{S_0}S_T\in\cA\right\}\,.
\end{equation}
The asset~$S$ will be called the \textit{eligible}, or \textit{reference},
\textit{asset}.
\end{definition}

\medskip

When finite and positive, $\rho_{\cA,S}(X)$ represents the ``minimum'' amount of
capital that needs to be invested in the eligible asset and added to the position~$X$
to reach acceptability. When negative, it represents the amount of capital that can
be extracted from~$X$ without compromising its acceptability. Clearly, unless~$\cA$
is closed, the infimum in~\eqref{eq def cap req} is not necessarily attained.

\medskip

Before stating some natural properties of risk measures~$\rho_{\cA,S}$, for which we
also refer to~\cite{FarkasKochMunari2012}, we recall some notation and terminology
for a map $\rho:\cX\to\overline{\R}$. The \textit{(effective) domain} of~$\rho$ is
the set
\begin{equation*}
\Dom(\rho):=\{X\in\cX \,; \ \rho(X)<\infty\}\,.
\end{equation*}
If the epigraph $\Epi(\rho):=\{(X,\alpha)\in\cX\times\R \,; \ \rho(X)\leq\alpha\}$ is
convex, respectively conic, then~$\rho$ is called \textit{convex}, respectively
\textit{positively homogeneous}. The function~$\rho$ is \textit{decreasing} if
$\rho(X)\geq\rho(Y)$ for all $X\le Y$. Moreover, we say that~$\rho$ is \textit{lower
semicontinuous} at~$X\in\cX$, if for every $\e>0$ there exists a neighborhood~$\cU$
of~$X$ such that $\rho(Y)\geq\rho(X)-\e$ for all $Y\in\cU$, and \textit{upper
semicontinuous} at~$X$ when~$-\rho$ is lower semicontinuous at~$X$. Note that
continuity is equivalent to having both lower and upper semicontinuity. If
$S=(S_0,S_T)$ is a traded asset the function~$\rho$ is said to be
\textit{$S$-additive} if for any $X\in\cX$
\begin{equation*}
\rho(X+\lambda S_T)=\rho(X)-\lambda S_0 \ \ \ \ \mbox{for all} \ \lambda\in\R\,.
\end{equation*}

\medskip

\begin{lemma}
\label{translation and monotonicity}
Let~$\cA\subset\cX$ be a monotone set and $S=(S_0,S_T)$ a traded asset.
Then~$\rho_{\cA,S}$ satisfies the following properties:
\begin{enumerate}[(i)]
  \item $\rho_{\cA,S}$ is $S$-additive and decreasing;
  \item $\{X\in\cX \,; \ \rho_{\cA,S}(X)=0\}\subset\partial\cA$ and
  \begin{equation}
  \label{chain of inclusions}
  \Interior(\cA)\subset\{X\in\cX \,; \
  \rho_{\cA,S}(X)<0\}\subset\cA\subset\{X\in\cX \,; \
  \rho_{\cA,S}(X)\leq0\}\subset\overline{\cA}\,;
  \end{equation}
  \item $\rho_{\cA,S}$ is lower semicontinuous at $X$ if and only if
      $X+\frac{m}{S_0}S_T\notin\overline{\cA}$ for any $m<\rho_{\cA,S}(X)$;
  \item $\rho_{\cA,S}$ is upper semicontinuous at $X$ if and only if
      $X+\frac{m}{S_0}S_T\in\Interior(\cA)$ for any $m>\rho_{\cA,S}(X)$;
  \item if~$\cA$ is convex, respectively conic, then~$\rho_{\cA,S}$ is convex,
      respectively positively homogeneous.
\end{enumerate}
\end{lemma}

\medskip

\begin{remark}
\label{KainaRuschendorf}
\begin{enumerate}[(i)]
  \item By part~\textit{(iv)} in the above lemma, the first inclusion
      in~\eqref{chain of inclusions} is an equality if and only if $\rho_{\cA,S}$
      is globally upper semicontinuous. By part~\textit{(iii)}, the last inclusion
      in~\eqref{chain of inclusions} is an equality if and only if $\rho_{\cA,S}$
      is globally lower semicontinuous.
  \item Note that $\rho_{\cA,S}$ cannot be (upper semi)continuous at any point $X$
      of its domain if $\Interior(\cA)$ is empty. This follows from
      part~\textit{(iv)} of the lemma above.
  \item Consider the space~$L^p$ for some $1\leq p<\infty$. By the previous point,
      Theorem~2.9 in~\cite{KainaRuschendorf2009} cannot be true in the stated
      generality, namely that any lower semicontinuous, coherent cash-additive risk
      measure $\rho:L^p\to\R\cup\{\infty\}$ must automatically be finitely valued
      and continuous. To see this, consider the closed, coherent acceptance
      set~$L^p_+$ and the risk-free asset $S=(1,1_\Omega)$. The corresponding risk
      measure $\rho_{L^p_+,S}$ is cash-additive, convex, and lower semicontinuous,
      but $\rho_{L^p_+,S}(X)=\infty$ whenever~$X$ is not essentially bounded from
      below. Moreover,~$\rho_{L^p_+,S}$ cannot be continuous at any point of
      finiteness since~$L^p_+$ has empty interior. The problem
      in~\cite{KainaRuschendorf2009} originates with the proof of Proposition~2.8
      in that paper which only works for finitely-valued functions.
\end{enumerate}
\end{remark}


\section{Interplay between the acceptance set and the eligible asset}
\label{sec 3}

In this section we investigate finiteness and continuity properties of risk measures
on general ordered topological vector spaces, highlighting the interplay between the
acceptance set and the eligible asset. Essentially, the more we require from the
acceptance set, the less we need to require from the  eligible asset.


\subsection{General acceptance sets}

Assume $\cA\subset L^\infty$ is an arbitrary acceptance set, and that the
payoff~$S_T$ of a traded asset~$S=(S_0,S_T)$ is an interior point of~$L^\infty_+$,
i.e.~$S_T$ is bounded away from zero. In this case, a standard argument shows that
the corresponding risk measure~$\rho_{\cA,S}$ is finitely valued and continuous, see
also~\cite{FarkasKochMunari2012}. For a general ordered topological vector
space~$\cX$ the statement remains true. When the interior of the positive cone is
empty we can still obtain finiteness if we require that~$S_T$ is an order unit.

\medskip

\begin{proposition}
\label{pseudo order unit}
Let $\cA\subset\cX$ be an arbitrary acceptance set and~$S=(S_0,S_T)$ a traded asset.
\begin{enumerate}[(i)]
\item If $S_T\in\Core(\cX_+)$, then~$\rho_{\cA,S}$ is finitely valued.
\item If $S_T\in\Interior(\cX_+)$, then~$\rho_{\cA,S}$ is finitely valued and
    continuous.
\end{enumerate}
Moreover, if~$\cX$ is an ordered normed space and $S_T\in\Interior(\cX_+)$,
then~$\rho_{\cA,S}$ is Lipschitz continuous.
\end{proposition}
\begin{proof}
\textit{(i)} Fix $X\in\cX$ and take $Y\in\cA$ and $Z\in\cA^c$. Since~$S_T$ belongs to
the core of~$\cX_+$, there exists $\lambda_1>0$ such that $Y-X\leq \lambda_1S_T$. As
a result, we have $X+\lambda_1S_T\in\cA$, implying $\rho_{\cA,S}(X)<\infty$. On the
other hand, we can also find $\lambda_2>0$ so that $X-Z\leq \lambda_2S_T$. Thus,
$X-\lambda_2S_T\notin\cA$ by monotonicity, showing that $\rho_{\cA,S}(X)>-\infty$.

\smallskip

\textit{(ii)} Since~$S_T$ is also an element of the core of~$\cX_+$, finiteness
follows from~\textit{(i)}. To prove continuity take an arbitrary $X\in\cX$ and assume
it is the limit of a net~$(X_\alpha)$. Since $\{Y\in\cX \,; \ -S_T\leq Y\leq S_T\}$
is a neighborhood of zero, for every $\e>0$ there exists~$\alpha_\e$ such that $-\e
S_T\leq X_\alpha-X\leq\e S_T$ whenever $\alpha\geq\alpha_\e$. But then
$\left|\rho_{\cA,S}(X_\alpha)-\rho_{\cA,S}(X)\right|\leq\e S_0$ for
$\alpha\geq\alpha_\e$, showing that~$\rho_{\cA,S}$ is continuous at~$X$.

\smallskip

Finally, assume~$\cX$ is an ordered normed space and $S_T\in\Interior(\cX_+)$ so
that~$\rho_{\cA,S}$ is finitely valued by part~\textit{(i)}. Using Theorem~9.40
in~\cite{AliprantisBorder2006} it is not difficult to prove that
$S_T\in\Interior(\cX_+)$ is equivalent to the existence of a constant $\lambda>0$
such that $X\leq\lambda\left\|X\right\|S_T$ for every nonzero $X\in\cX$. To prove
Lipschitz continuity, take now two positions~$X$ and~$Y$ in~$\cX$. Since $Y\leq
X+\lambda\left\|X-Y\right\|S_T$, we obtain that
$\rho_{\cA,S}(X)-\rho_{\cA,S}(Y)\leq\lambda S_0\left\|X-Y\right\|$. Exchanging~$X$
and~$Y$, we conclude the proof.
\end{proof}

\medskip

\begin{remark}
\label{pre-order remark}
The above proposition is easily seen to hold if $\cX_+$ is only assumed to induce a
pre-ordering on $\cX$, i.e. if $\cX_+$ is a convex cone which is not necessarily
pointed. This will be important in Theorem~\ref{finite coherent} where the
proposition is applied with respect to the pre-ordering induced by a coherent
acceptance set.
\end{remark}

\medskip

We now turn to general acceptance sets in topological Riesz spaces, i.e. topological
vector spaces equipped with a lattice ordering. First we need the following
generalization of the notion of order units. An element $Z\in\cX_+$ in a Riesz
space~$\cX$ is called a \textit{weak topological unit} if for every $X\in\cX_+$ we
have $X\wedge nZ\to X$ as $n\to\infty$.

\medskip

The next technical lemma extends Theorem~6.3 in~\cite{Schaefer1974} outside the
normed space setting and establishes the link between weak topological units and
strictly positive elements. Recall that a topological Riesz space $\cX$ is said to be
locally solid when there exists a neighborhood base of zero consisting of
\textit{solid} neighborhoods~$\cU$, i.e. satisfying $X\in\cU$ whenever $Y\in\cU$ and
$\left|X\right|\leq\left|Y\right|$. For more details, we refer to Chapter~9
in~\cite{AliprantisBorder2006}.

\medskip

\begin{lemma}
\label{equivalence q.i. point}
Let~$\cX$ be a topological Riesz space. Then every weak topological unit is strictly
positive. If~$\cX$ is also locally convex and locally solid, the converse is true as
well.
\end{lemma}
\begin{proof}
Let~$Z$ be a weak topological unit and assume $\psi(Z)=0$ for some nonzero
$\psi\in\cX'_+$. Then $\psi(X\wedge nZ)=0$ for all $X\in\cX_+$ and all positive
integers~$n$. Hence, by continuity, $\psi(X)=0$ for all $X\in\cX_+$, implying
that~$\psi$ is null. This proves~$Z$ must be strictly positive.

\smallskip

Assume now that~$\cX$ is locally convex. Let~$Z$ be strictly positive and take
$X\in\cX_+$. To prove that $Z$ is a weak topological unit, it is sufficient to show
that for every solid neighborhood of zero~$\cU$  we eventually have $X-(X\wedge
nZ)\in\cU$. By Theorem~8.54 in~\cite{AliprantisBorder2006} the principal ideal
\begin{equation*}
\ccI_Z:=\{Y\in\cX \,; \ \exists \ \lambda>0 \,:\, \left|Y\right|\leq\lambda Z\}
\end{equation*}
is
weakly dense in~$\cX$. Since~$\ccI_Z$ is convex and~$\cX$ is locally convex, this
implies that~$\ccI_Z$ is dense in~$\cX$ with respect to the original topology. As a
result, we can find $Y\in\ccI_Z$ with $X-Y\in\cU$. Setting $W:=X\wedge Y^+$ and
noting that~$W$ belongs to~$\ccI_Z$, we see that $W\leq n_0 Z$ for some positive
integer~$n_0$. Since for all $n\geq n_0$
\begin{equation*}
0 \leq X-(X\wedge nZ) \leq X-(X\wedge n_0Z) \leq X-W \leq X\vee Y-X\wedge Y =
\left|X-Y\right|\,,
\end{equation*}
the solidity of~$\cU$ implies that $X-(X\wedge nZ)\in\cU$ for every $n\geq n_0$,
concluding the proof.
\end{proof}

\medskip

\begin{remark}
\begin{enumerate}[(i)]
\item Weak \textit{topological} units differ from weak \textit{order} units
    $Z\in\cX_+$ which for all $X\in\cX_+$ satisfy $X=\sup_{n}X\wedge nZ$. For
    instance, every element $Z\in L^\infty_+$ with $Z>0$ almost surely is a weak
    order unit, but not a weak topological unit unless it is bounded away from
    zero.

\item By the previous result, weak topological units in~$L^p$ spaces, $1\leq
    p<\infty$, or in Orlicz hearts, are precisely those positive elements~$Z$ for
    which $Z>0$ almost surely. In~$L^\infty$ they correspond to elements that are
    bounded away from zero.
\item Recall that~$L^p$ is a topological Riesz space which is not locally convex
    whenever $0\leq p<1$. In this case, the set of strictly positive elements
    coincides with the positive cone since the only continuous linear functional is
    the zero functional. However, it is not difficult to show that $Z\in L^p_+$ is
    a weak topological unit if and only if $Z>0$ almost surely.
\end{enumerate}
\end{remark}

\medskip

The next theorem is the main result of this section and provides a sufficient
condition for a risk measure on a topological Riesz space to be finitely valued. We
require neither convexity of $\cA$ nor cash additivity of $\rho_{\cA,S}$. Our result
contains as a special case nonconvex extensions of two well-known finiteness results
for convex cash-additive risk measures: Theorem~2.3 in Svindland~\cite{Svindland2008}
on~$L^p$ spaces, $1\leq p\leq\infty$, and Theorem~4.6 in Cheridito and
Li~\cite{CheriditoLi2009} on Orlicz hearts. The proofs of both of these results rely
on separation arguments which cannot be reproduced in our nonconvex setting. In fact,
our approach is simpler and depends solely on the lattice structure. It is closer in
spirit to the proof of Proposition~6.7 in Shapiro, Dentcheva and
Ruszczy\'{n}ski~\cite{DentchevaRuszczynskiShapiro2009}, who, however, make use of a
category argument that only works if lower semicontinuity is additionally assumed.

\medskip

\begin{theorem}
\label{finiteness in frechet lattice}
Let~$\cX$ be a topological Riesz space, and $\cA\subset\cX$ an acceptance set with
nonempty interior. Let $S=(S_0,S_T)$ be a traded asset and assume that~$\rho_{\cA,S}$
does not attain the value~$-\infty$. If~$S_T$ is a weak topological unit,
then~$\rho_{\cA,S}$ is finitely valued.
\end{theorem}
\begin{proof}
Take $Z\in\Interior(\cA)$ and choose a neighborhood of zero~$\cU$ such that
$Z+\cU\subset\cA$. Fix $Y\in\cX_+$ and note that $Y=Y\wedge(nS_T)+(Y-nS_T)^+$ for any $n\in\N$. Since~$S_T$ is a weak topological unit, we have
$(Y-nS_T)^+\to0$ as $n\to\infty$, so that $-(Y-mS_T)^+\in\cU$ for a sufficiently large~$m$. Note that
$Z-(Y-mS_T)^+\in\cA$ and $Z-(Y-mS_T)^+-mS_T\leq Z-Y$. Hence, by monotonicity and by $S$-additivity, we have $\rho_{\cA,S}(Z-Y)\leq mS_0<\infty$. Now take an arbitrary
$X\in\cX$. Setting $Y:=(Z-X)^+$, it follows that
$\rho_{\cA,S}(X)\leq\rho_{\cA,S}(Z-Y)<\infty$. Hence~$\rho_{\cA,S}$ is finitely
valued.
\end{proof}

\medskip

When~$\cX$ is a Fr\'{e}chet lattice, i.e. a topological Riesz space which is locally
solid and completely metrizable, the interior and the core of a monotone set always
coincide. This can be shown by adapting the proof of Lemma~4.1 for monotone
functionals on a Banach lattice in Cheridito and Li~\cite{CheriditoLi2009}.
Consequently, on Fr\'{e}chet lattices the above theorem holds under the weaker
assumption that the acceptance set has nonempty core. Because of its practical
relevance -- it is generally easier to show that an element belongs to the core than
to show it belongs to the interior of a set -- we record it in the next proposition.

\medskip

\begin{proposition}
\label{interiors of acceptance}
Let~$\cX$ be a Fr\'{e}chet lattice. The following statements hold:
\begin{enumerate}[(i)]
  \item $\Interior(\cA)=\Core(\cA)$ for every monotone set $\cA\subset\cX$;
  \item if $\cA$ is an acceptance set with nonempty core and $S=(S_0,S_T)$ a traded
      asset with $S_T$ weak topological unit, then~$\rho_{\cA,S}$ is finitely
      valued whenever it does not attain the value $-\infty$.
\end{enumerate}
\end{proposition}

\medskip

\begin{remark}
\begin{enumerate}[(i)]
	\item Part~\textit{(i)} of the above proposition provides an alternative approach
to the extended Namioka-Klee theorem obtained in Biagini and
Frittelli~\cite{BiaginiFrittelli2009}: \textit{Every convex monotone function
$\rho:\cX\to\R\cup\{\infty\}$ on a Fr\'{e}chet lattice $\cX$ is continuous on the
interior of its domain}. Indeed, assume~$\rho$ is such a map and let $X$ be an
interior point of its domain. As in the proof of Proposition~\ref{lemma on internal
point and domain} below, it is not difficult to show that, for any
$\alpha>\rho(X)$, the point~$X$ belongs to the core of $\cA:=\{Y\in\cX \,; \
\rho(Y)<\alpha\}$. Then~$X$ is an interior point of~$\cA$ by
Proposition~\ref{interiors of acceptance}, hence the map~$\rho$ turns out to be
bounded from above on a neighborhood of~$X$. As a result, Theorem~5.43
in~\cite{AliprantisBorder2006} implies~$\rho$ is continuous at~$X$.
\item If the acceptance set in the preceding proposition is additionally assumed to be
    convex, the finiteness of~$\rho_{\cA,S}$ immediately implies continuity by
    Theorem~1 in~\cite{BiaginiFrittelli2009} or by the first remark.
\end{enumerate}
\end{remark}


\subsection{Convex acceptance sets}

In this section we focus on convex acceptance sets and provide a variety of
finiteness and continuity results in general ordered topological vector spaces.
Convexity allows us to obtain results for a wide range of eligible assets, without
requiring that the positive cone has nonempty interior. In particular, all results in
this section apply to~$L^p$, $1\leq p\leq\infty$, and Orlicz spaces. We start by
showing a general necessary condition for a convex risk measure to be finite.

\medskip

\begin{proposition}
\label{lemma on internal point and domain}
Let $\cA\subset\cX$ be a convex acceptance set and $S=(S_0,S_T)$ a traded asset.
Assume that~$\rho_{\cA,S}$ does not attain the value~$-\infty$. Then
$\Core(\Dom(\rho_{\cA,S}))$ is nonempty if and only if~$\Core(\cA)$ is nonempty.

\smallskip

In particular, if~$\rho_{\cA,S}$ is finitely valued then~$\Core(\cA)$ is nonempty.
\end{proposition}
\begin{proof}
Since $\cA\subset\Dom(\rho_{\cA,S})$, it is enough to prove the ``only if'' part. Let
$X\in\Core(\Dom(\rho_{\cA,S}))$ and assume without loss of generality that
$\rho_{\cA,S}(X)<0$. Take a nonzero $Y\in\cX$ and choose $\e>0$ in such a way that
$X+\lambda Y\in\Dom(\rho_{\cA,S})$ whenever $\lambda\in (-\e,\e)$. Then
$f(\lambda):=\rho_{\cA,S}(X+\lambda Y)$ defines a real-valued function on~$(-\e,\e)$,
which must be continuous by convexity. Since $f(0)=\rho_{\cA,S}(X)<0$, it follows
that there exists $\delta>0$ such that $\rho_{\cA,S}(X+\lambda Y)=f(\lambda)<0$ for
$\lambda\in(-\delta,\delta)$ and, consequently, $X+\lambda Y\in\cA$ for all
such~$\lambda$. In conclusion,~$X\in\Core(\cA)$.
\end{proof}

\medskip

\begin{remark}
If~$\cX$ is a Fr\'{e}chet lattice and $\cA\subset\cX$ a convex acceptance set, it
follows immediately from the above result and Proposition~\ref{interiors of
acceptance} that the domain of a risk measure~$\rho_{\cA,S}$ has nonempty interior if
and only if~$\cA$ itself has nonempty interior.
\end{remark}

\medskip

The preceding remark allows us to reformulate the continuity part of Theorem~1 in
Biagini and Frittelli~\cite{BiaginiFrittelli2009} when restricted to convex risk
measures as follows: \textit{Let~$\cX$ be a Frech\'{e}t lattice, $\cA\subset\cX$ a
convex acceptance set with nonempty interior and $S=(S_0,S_T)$ a traded asset.
If~$\rho_{\cA,S}$ does not assume the value~$-\infty$, then it is continuous on the
interior of its domain.} As a consequence, the following result can be regarded as an
extended Namioka-Klee theorem for convex risk measures defined on general ordered
topological vector spaces. Note that no lattice structure is required here and the
proof is more direct.

\medskip

\begin{theorem}
\label{cont after finit}
Let $\cA\subset\cX$ be a convex acceptance set, and~$S=(S_0,S_T)$ a traded asset.
Assume~$\rho_{\cA,S}$ does not take the value~$-\infty$. The following statements are
equivalent:
\begin{enumerate}[(a)]
	\item $\Dom(\rho_{\cA,S})$ has nonempty interior and $\rho_{\cA,S}$ is
continuous on $\Interior(\Dom(\rho_{\cA,S}))$;
	\item $\Interior(\cA)$ is nonempty.
\end{enumerate}
In particular, if~$\cA$ has nonempty interior then~$\rho_{\cA,S}$ is continuous
on~$\cX$ whenever it is finitely valued.
\end{theorem}
\begin{proof}
By Remark~\ref{KainaRuschendorf} it is enough to prove that~\textit{(b)}
implies~\textit{(a)}. Note first that the domain of~$\rho_{\cA,S}$ has nonempty
interior because it contains~$\cA$. Since~$\rho_{\cA,S}$ is bounded above by~$0$
on~$\Interior(\cA)$, we can apply Theorem~5.43 in~\cite{AliprantisBorder2006} to
obtain~\textit{(a)}.
\end{proof}

\medskip

We now focus on finiteness results in the context of convex acceptance sets with
nonempty interior. In this case, finiteness always implies continuity by
Theorem~\ref{cont after finit}. The following lemma will prove to be useful.

\medskip

\begin{lemma}
\label{halfspaces containing acceptance sets}
Let $\cA\subset\cX$ be an arbitrary acceptance set and consider a (not necessarily
continuous) linear functional $\psi:\cX\to\R$. Then~$\psi$ is positive whenever it is
bounded from below on~$\cA$.
\end{lemma}
\begin{proof}
Let $X\in\cX_+$ be arbitrary and fix $Y\in\cA$. Then, by monotonicity of~$\cA$, we
have $Y+\lambda X\in\cA$ for all $\lambda\geq0$. Hence,
$\psi(Y)+\lambda\psi(X)\geq\inf_{Z\in\cA}\psi(Z)>-\infty$ for all $\lambda\geq0$,
which can only be true if $\psi(X)\geq0$.
\end{proof}

\medskip

We start by showing that if a risk measure is finitely valued in the direction of
some strictly positive element, then it is finitely valued on~$\cX$. This provides a
simple criterion for finiteness and continuity which we will use in
Proposition~\ref{risk measure u finite} in the context of shortfall risk measures.
Note that we do not require any explicit assumption on the eligible asset~$S$.

\medskip

\begin{theorem}
\label{general finiteness convex with str pos elements}
Assume~$\cX$ admits a strictly positive element $U\in\cX_+$. Let $\cA\subset\cX$ be a
convex acceptance set with nonempty interior, and~$S=(S_0,S_T)$ a traded asset.
Assume~$\rho_{\cA,S}$ does not attain the value~$-\infty$. Then~$\rho_{\cA,S}$ is
finitely valued if and only if $\rho_{\cA,S}(-\lambda U)<\infty$ for all $\lambda>0$.
In this case,~$\rho_{\cA,S}$ is also continuous.
\end{theorem}
\begin{proof}
We only need to prove the ``if'' part. Assume $X\notin\Dom(\rho_{\cA,S})$. Since
$\Dom(\rho_{\cA,S})$ is convex and has nonempty interior, by separation and
Lemma~\ref{halfspaces containing acceptance sets} we find a nonzero $\psi\in\cX'_+$
with $\psi(X)\leq\psi(-\lambda U)$ for all $\lambda>0$. But this implies $\psi(U)=0$,
contradicting the strict positivity of~$U$. Hence,~$\rho_{\cA,S}$ must be finitely
valued and, hence, also continuous.
\end{proof}

\medskip

\begin{remark}
\begin{enumerate}[(i)]
\item The theorem above is particularly useful when $\cX$ is an $L^p$ space, $1\le
    p\le\infty$, or an Orlicz heart, since $U:=1_\Omega$ is a strictly positive
    element in these spaces
\item Note that if the acceptance set in the preceding theorem is assumed to be
    coherent, the condition $\rho_{\cA,S}(-\lambda U)<\infty$ for all $\lambda>0$
    becomes equivalent to $\rho_{\cA,S}(-U)<\infty$ due to positive homogeneity.
\end{enumerate}
\end{remark}

\medskip

By Proposition~\ref{pseudo order unit}, for general acceptance sets we always have
finiteness if the payoff of the eligible asset is an order unit. If the acceptance
set is convex and has nonempty interior it suffices to require the payoff to be strictly positive. In Proposition~\ref{risk measure u finite} below we will show that this condition is sometimes also necessary for finiteness. Note that, in contrast to other results in this section, we do not need to require a priori that the risk measure does not attain the value~$-\infty$.

\medskip

\begin{corollary}
\label{finiteness for convex acc}
Let $\cA\subset\cX$ be a convex acceptance set with nonempty interior, and
$S=(S_0,S_T)$ a traded asset. If~$S_T$ is strictly positive, then~$\rho_{\cA,S}$ is
finitely valued and continuous.
\end{corollary}
\begin{proof}
First, we show that~$\rho_{\cA,S}$ never attains the value~$-\infty$. Indeed, assume
to the contrary that $\rho_{\cA,S}(X)=-\infty$ for some $X\in\cX$, and take
$Y\notin\cA$. By a standard separation argument, there exists a nonzero $\psi\in\cX'$
such that $\psi(Y)\leq\psi(X+\lambda S_T)$ for any $\lambda\in\R$. Hence, we must
have $\psi(S_T)=0$. Note that $\psi$ is positive due to Lemma~\ref{halfspaces
containing acceptance sets}. Since $S_T$ is strictly positive, $\psi(S_T)=0$ cannot
hold and we conclude that~$\rho_{\cA,S}$ does not attain the value~$-\infty$.

\smallskip

By Theorem~\ref{general finiteness convex with str pos elements} and $S$-additivity,
to conclude the proof we just need to show that $\rho_{\cA,S}(0)<\infty$. If this is
not the case, then $\R S_T\cap\cA=\emptyset$. Thus we can find a nonzero separating
functional $\varphi\in\cX'$ such that
$\lambda\varphi(S_T)\leq\varphi(X)$ for every $X\in\cA$ and $\lambda\in\R$. This
implies $\varphi(S_T)=0$, which is again in contrast to the positivity of~$\varphi$
ensured by Lemma~\ref{halfspaces containing acceptance sets}. Hence
$\rho_{\cA,S}(0)<\infty$, concluding the proof.
\end{proof}

\medskip

We now show that, when the underlying acceptance set has nonempty interior, a convex
risk measure which is finitely valued on a dense subspace is automatically finitely
valued on the whole space. This is particularly useful when dealing with risk
measures defined on~$L^p$, $1\leq p<\infty$, or on Orlicz hearts~$H^\Phi$ since,
typically, it is not difficult to establish finiteness on the dense
subspace~$L^\infty$. The result is also valid for general convex maps whose domain
has nonempty interior.

\medskip

\begin{proposition}
\label{finiteness and density}
Let $\cA\subset\cX$ be a convex acceptance set with nonempty interior, and
$S=(S_0,S_T)$ a traded asset. Assume~$\rho_{\cA,S}$ does not attain the
value~$-\infty$. If~$\rho_{\cA,S}$ is finitely valued on a dense linear
subspace~$\cS$ of~$\cX$, then~$\rho_{\cA,S}$ is finitely valued and continuous
on~$\cX$.
\end{proposition}
\begin{proof}
Assume $X\notin\Dom(\rho_{\cA,S})$. Since the domain of~$\rho_{\cA,S}$ is convex and
contains~$\cA$, by separation we find a nonzero $\psi\in\cX'$ such that
$\psi(X)\leq\psi(Y)$ for all $Y\in\cS$. But this implies~$\psi$ must annihilate~$\cS$
and hence, by density, the whole space~$\cX$. Therefore,~$\rho_{\cA,S}$ must be
finitely valued, hence continuous, on the whole of~$\cX$.
\end{proof}


\subsection{Conic and coherent acceptance sets}
\label{section conic coherent}

In this section, we focus our analysis on conic and coherent acceptance sets. We
start with the main result characterizing the range of eligible assets for which a
coherent risk measure is finitely valued, respectively continuous. We will apply this
result to risk measures based on $\TVaR$-acceptability in~$L^p$ spaces in
Section~\ref{applications tvar}. Note that, if $\cA\subset\cX$ is a coherent
acceptance set, the relation $X\leq_\cA Y$ defined by $Y-X\in\cA$ is a pre-ordering
on~$\cX$ with positive cone~$\cA$. Note that $\leq_\cA$ is not an ordering
unless~$\cA$ is pointed.

\medskip

\begin{theorem}
\label{finite coherent}
Assume $\cA\subset\cX$ is a coherent acceptance set and let $S=(S_0,S_T)$ be a traded
asset.
\begin{enumerate}[(i)]
\item The following statements are equivalent:
\begin{enumerate}[(a)]
	\item $\rho_{\cA,S}$ is finitely valued;
	\item $S_T\in\Core(\cA)$.
\end{enumerate}
\item The following statements are equivalent:
\begin{enumerate}[(a)]
	\item $\rho_{\cA,S}$ is continuous on~$\cX$;
	\item $\rho_{\cA,S}$ is continuous at~$0$;
	\item $S_T\in\Interior(\cA)$.
\end{enumerate}
\end{enumerate}
Moreover, if $\cX$ is an ordered normed space, then~$\rho_{\cA,S}$ is Lipschitz
continuous whenever $S_T\in\Interior(\cA)$.
\end{theorem}
\begin{proof}
\textit{(i)} Note that~\textit{(b)} is equivalent to~$S_T$ being an order unit with
respect to~$\leq_\cA$. Hence,~\textit{(b)} implies~\textit{(a)} by
Remark~\ref{pre-order remark}. To prove the converse assume~$\rho_{\cA,S}$ is
finitely valued but $S_T\notin\Core(\cA)$. Then we can find $X\in\cX$ such that
$S_T+\lambda_n X\notin\cA$ for a suitable sequence~$(\lambda_n)$ of strictly positive
numbers converging to zero. Equivalently, $X+\frac{1}{\lambda_n}S_T\notin\cA$ for
every~$n\in\N$, implying $\rho_{\cA,S}(X)=\infty$.

\smallskip

\textit{(ii)} Clearly~\textit{(a)} implies~\textit{(b)}. If~$\rho_{\cA,S}$ is
continuous at~$0$, then for every $m>0\geq\rho_{\cA,S}(0)$ we have
$\frac{m}{S_0}S_T\in\Interior(\cA)$ by Lemma~\ref{translation and monotonicity}.
Taking $m:=S_0$ we see that $S_T\in\Interior(\cA)$, proving that~\textit{(b)}
implies~\textit{(c)}. Finally, if~\textit{(c)} holds then we can again pass to the
induced pre-ordering~$\leq_\cA$ and refer to Remark~\ref{pre-order remark} to
conclude the proof.
\end{proof}

\medskip

We conclude this section with a finiteness result in the context of conic acceptance
sets, omitting the easy proof. This result will be of practical importance in the
context of risk measures based on $\VaR$-acceptability in~$L^p$ spaces treated
in Section~\ref{applications var}.

\medskip

\begin{proposition}
\label{compact notation cones}
Assume $\cA\subset\cX$ is a conic acceptance set and let $S=(S_0,S_T)$ be a traded
asset. The following statements hold:
\begin{enumerate}[(i)]
  \item $\rho_{\cA,S}<\infty$ if and only if $S_T\in\Core(\cA)$;
	\item $\rho_{\cA,S}>-\infty$ if and only if $-S_T\in\Core(\cA^c)$.
\end{enumerate}
In particular, if~$\rho_{\cA,S}$ is finitely valued then~$\Core(\cA)$ is nonempty.
\end{proposition}


\section{Applications}
\label{section applications}

We now apply our previous results to provide complete characterizations of finiteness
and continuity for risk measures on~$L^p$ spaces based on the two most prominent
acceptability criteria in practice: Value-at-Risk and Tail Value-at-Risk. We also
provide a treatment of shortfall risk measures on Orlicz spaces arising from utility
functions. Throughout this entire section we maintain the assumption that
$(\Omega,\cF,\probp)$ is a nonatomic probability space.


\subsection{Acceptability based on Value-at-Risk}
\label{applications var}

In this subsection we work in the setting of~$\cX=L^p$ for a fixed $0\leq p<\infty$.
The case $p=\infty$ is analogous to the case where~$\cX$ is the space of bounded
measurable functions for which we refer to~\cite{FarkasKochMunari2012}.

\medskip

For $\alpha\in(0,1)$ the {\em Value-at-Risk} of $X\in L^p$ at the level~$\alpha$ is
defined as
\begin{equation*}
\label{def var}
\VaR_\alpha(X):=\inf\{m\in\R \,; \ \probp(X+m<0)\le\alpha\}\,.
\end{equation*}

\medskip

The set
\begin{equation*}
\cA_{\alpha}:=\{X\in L^p \,; \ \VaR_\alpha(X)\le 0\}=\{X\in L^p \,; \
\probp(X<0)\le\alpha\}
\end{equation*}
is a conic acceptance set which is not convex and which is well known to be closed,
see for instance Theorem~3 in~\cite{Chambers2009}. The following lemma describes the
interior of~$\cA_\alpha$.

\medskip

\begin{lemma}
\label{lemma core var}
The acceptance set~$\cA_\alpha$ has nonempty interior in~$L^p$. Moreover,
\begin{equation}
\label{core var Lp}
\Interior(\cA_\alpha)=\{X\in L^p \,; \ \probp(X\leq0)<\alpha\}\,.
\end{equation}
In particular, for $S_T\in L^p_+$ we have $S_T\in\Interior(\cA_\alpha)$ if and only
if $\probp(S_T=0)<\alpha$.
\end{lemma}
\begin{proof}
To prove~\eqref{core var Lp} first recall that, by Proposition~\ref{interiors of
acceptance}, the core and the interior of any acceptance set in~$L^p$ coincide. Take
now $X\in L^p$ with $\probp(X\leq0)<\alpha$. If $X\notin\Core(\cA_\alpha)$, then we
can find $Z\in L^p_+$ and $\lambda_n\downarrow 0$ such that $\probp(X<\lambda_n
Z)>\alpha$, implying $\probp(X\leq0)\geq\alpha$. But this contradicts what assumed
above, hence~$X$ must belong to~$\Core(\cA_\alpha)$.

\smallskip

To prove the converse inclusion take $X\in\Core(\cA_ \alpha)$ and assume
$\probp(X\leq0)\geq\alpha$. Since $X\in\cA_\alpha$, we have $\probp(X>0)>0$ and thus
$\probp(0<X<\e)>0$ for some $\e>0$. Therefore, we find a sequence~$(A_n)$ of pairwise
disjoint measurable subsets of $\{0<X<\e\}$ with $0<\probp(A_n)<n^{-p-2}$. Setting
$Z:=1_{\{X\leq 0\}}+\sum_n n1_{A_n}\in L^p_+$ it is easy to see that for every
$\lambda>0$ there exists a positive integer~$n$ for which $\probp(X<\lambda
Z)\geq\probp(X\leq0)+\probp(A_n)>\alpha$. But this contradicts
$X\in\Core(\cA_\alpha)$, hence~\eqref{core var Lp} must hold.
\end{proof}

\medskip

Let $S=(S_0,S_T)$ be a traded asset. The corresponding risk measure based on
$\VaR$-acceptability is
\begin{equation*}
\rho_{\cA_\alpha,S}(X)=\inf\left\{m\in\R \,; \
\probp\left(X+\frac{m}{S_0}S_T<0\right)\leq\alpha\right\}\,.
\end{equation*}

\medskip

The following proposition provides a characterization of the finiteness
of~$\rho_{\cA_\alpha,S}$ and shows that risk measures based on $\VaR$-acceptability
can never be globally continuous, regardless of the choice of the eligible asset.

\medskip

\begin{proposition}
\label{finiteness and cont var}
Let $S=(S_0,S_T)$ be a traded asset. The following statements are equivalent:
\begin{enumerate}[(a)]
	\item $\rho_{\cA_\alpha,S}$ is finitely valued on~$L^p$;
	\item $\probp(S_T=0)<\min\{\alpha,1-\alpha\}$.
\end{enumerate}
Moreover,~$\rho_{\cA_\alpha,S}$ is never globally continuous on~$L^p$.
\end{proposition}
\begin{proof}
To characterize finiteness, by Proposition~\ref{compact notation cones} and
Lemma~\ref{lemma core var} we only need to show that~$\rho_{\cA_\alpha,S}$ never
attains the value~$-\infty$ if and only if $\probp(S_T=0)<1-\alpha$. If
$\probp(S_T=0)\ge 1-\alpha$, then clearly $\rho_{\cA_\alpha,S}(0)=-\infty$. On the
other side, assume $\probp(S_T=0)<1-\alpha$. Note that for any $X\in L^p$ we have
$\probp(\{X<nS_T\}\cap\{S_T>0\})\to\probp(S_T>0)$ as $n\to\infty$, implying
$\probp(X<nS_T)>\alpha$ for large enough~$n\in\N$. Hence it follows that
$\rho_{\cA_\alpha,S}(X)>-\infty$.

\smallskip

To show that~$\rho_{\cA_\alpha,S}$ is never continuous on the whole~$L^p$, take
$\e>0$ and a measurable set with $\probp(A)=\alpha$, and set $X:=-(S_T+\e)1_A\in
L^p$. Note that $\probp(X<0)=\alpha$ and $\probp(X+S_T\leq0)\geq\alpha$. As a result,
$\rho_{\cA_\alpha,S}(X) \leq 0 < S_0 \leq \rho_{\Interior(\cA_\alpha),S}(X)$ and
Lemma~\ref{translation and monotonicity} implies that~$\rho_{\cA_\alpha,S}$ cannot be
(upper semi)continuous at~$X$.
\end{proof}


\subsection{Acceptability based on Tail Value-at-Risk}
\label{applications tvar}

We continue to work on~$\cX= L^p$ for a fixed $1\leq p<\infty$. As for Value-at-Risk,
the case $p=\infty$ can be treated similarly to the case of bounded measurable
functions which can be found in~\cite{FarkasKochMunari2012}.

\medskip

Fix $\alpha\in(0,1)$. The {\em Tail Value-at-Risk} of $X\in L^p$ at the
level~$\alpha$ is defined as
\begin{equation*}
\TVaR_\alpha(X):=\frac{1}{\alpha} \int_0^\alpha\VaR_\beta(X)d\beta\,.
\end{equation*}
It is well known that~$\TVaR_\alpha$ is cash additive and Lipschitz continuous
on~$L^1$ and, therefore, also on~$L^p$. Therefore, the set
\begin{equation*}
\cA^{\alpha}:=\{X\in L^p \,; \ \TVaR_\alpha(X)\le 0\}
\end{equation*}
is a closed, coherent acceptance set which has nonempty interior. Moreover, note that
$\cA^\alpha\subset\cA_{\alpha}$.

\medskip

\begin{lemma}
\label{lemma core tvar}
The following holds:
\begin{equation}
\label{core tvar Lp}
\Interior(\cA^\alpha)=\{X\in L^p \,; \ \TVaR_\alpha(X)<0\}\subset\{X\in L^p \,; \
\probp(X\leq0)<\alpha\}\,.
\end{equation}
For $S_T\in L^p_+$ we have $S_T\in\Interior(\cA^\alpha)$ if and only if
$\probp(S_T=0)<\alpha$.
\end{lemma}
\begin{proof}
The equality in~\eqref{core tvar Lp} follows from Remark~\ref{KainaRuschendorf}. Moreover, since $\cA^\alpha\subset\cA_\alpha$, the inclusion in~\eqref{core tvar Lp} is an immediate consequence of Lemma~\ref{lemma core var}. Finally, if $X\in L^p_+$ and $\probp(X=0)<\alpha$ we must find $\lambda>0$ in such a way that $\gamma:=\probp(X<\lambda)<\alpha$. Then $\VaR_\beta(X)<0$ for all $\beta\in(\gamma,\alpha)$. Since~$X$ is positive, this implies $\TVaR_\alpha(X)<0$.
\end{proof}

\medskip

Given a traded asset $S=(S_0,S_T)$, we consider the corresponding risk measure based
on $\TVaR$-acceptability
\begin{equation*}
\rho_{\cA^\alpha,S}(X)=\inf\left\{m\in\R \,; \
\TVaR_\alpha\left(X+\frac{m}{S_0}S_T\right)\leq0\right\}\,.
\end{equation*}

\medskip

The following proposition provides a characterization of the finiteness and
continuity of $\TVaR$-based risk measures and is a direct consequence of the above
lemma and the results in Section~\ref{section conic coherent}. Note the strong
contrast to $\VaR$-based risk measures, which are never globally continuous on~$L^p$
for $p<\infty$.

\medskip

\begin{proposition}
\label{finiteness and continuity tvar}
Let $S=(S_0,S_T)$ be a traded asset. The following statements are equivalent:
\begin{enumerate}[(a)]
	\item $\rho_{\cA^\alpha,S}$ is finitely valued on~$L^p$;
	\item $\rho_{\cA^\alpha,S}$ is Lipschitz continuous on~$L^p$;
	\item $\TVaR_\alpha(S_T)<0$;
	\item $\probp(S_T=0)<\alpha$.
\end{enumerate}
\end{proposition}


\subsection{Acceptability based on shortfall risk}
\label{applications utility}

Cash-additive risk measures based on utility functions have been widely investigated
on spaces of bounded measurable functions, see~\cite{FoellmerSchied2011} for a
general overview. As a means of unifying the treatment of utility maximization
problems, Biagini and Frittelli proposed in~\cite{BiaginiFrittelli2008} to work
instead in the setting of Orlicz spaces; see also Biagini and
Frittelli~\cite{BiaginiFrittelli2009} and Arai~\cite{Arai2011}.

\medskip

Recall that a nonconstant function $u:\R\to\R$ is a \textit{utility function} if it
is concave and increasing. Note that this implies
$u(-\infty):=\lim_{x\to-\infty}u(x)=-\infty$. The function defined by
\begin{equation*}
\label{orlicz from utility}
\widehat{u}(x):=u(0)-u(-\left|x\right|)
\end{equation*}
is an Orlicz function in the sense of Definition~2.1.1 in~\cite{EdgarSucheston1992}.
By~$H^{\widehat{u}}$ we denote the corresponding Orlicz heart associated
with~$(\Omega,\cF,\probp)$.

\medskip

We fix a level $\alpha\in\R$ such that $\alpha\leq u(x_0)$ for some $x_0\in\R$. Then
the set
\begin{equation*}
\cA_u:=\{X\in H^{\widehat{u}} \,; \ \E[u(X)]\geq\alpha\}
\end{equation*}
is a convex acceptance set which, in general, is not coherent. Note that we disregard
any level~$\alpha$ strictly bounding~$u$ from above, since then~$\cA_u$ would be
empty. If $S=(S_0,S_T)$ is a traded asset, the corresponding \textit{shortfall risk
measure} on~$H^{\widehat{u}}$ is defined by
\begin{equation*}
\rho_{\cA_u,S}(X)=\inf\left\{m\in\R \,; \
\E\left[u\left(X+\frac{m}{S_0}S_T\right)\right]\geq\alpha\right\}\,.
\end{equation*}

\medskip

We start by describing the topological properties of the acceptance set~$\cA_u$.

\medskip

\begin{lemma}
\label{charact nonempty int in orlicz}
\begin{enumerate}[(i)]
  \item The set $\cA_u$ has nonempty interior if and only if $u(x_0)>\alpha$ for
      some $x_0>0$.
	\item If~$u$ is bounded from above, then $\cA_u$ is closed.
\end{enumerate}
\end{lemma}
\begin{proof}
\textit{(i)} To prove the ``if'' part, we show that $X:=x_01_\Omega$ is an interior
point of~$\cA_u$. Choose $\lambda\in(0,1)$ in such a way that $\alpha-\lambda
u(x_0)+(1-\lambda)(1-u(0))\leq0$. Note that for every $Y\in H^{\widehat{u}}$ with
$\left\|Y\right\|_{\widehat{u}}<1-\lambda$ we have
$\E\left[\widehat{u}\left(\frac{Y}{1-\lambda}\right)\right]\leq1$, yielding
\begin{eqnarray*}
\E[-u(X+Y)]
&\leq&
\lambda\E\left[-u\left(\frac{X}{\lambda}\right)\right]+(1-\lambda)\E\left[-u\left(\frac{Y}{1-\lambda}\right)\right]
\\
&\leq&
-\lambda
u\left(\frac{x_0}{\lambda}\right)+(1-\lambda)\E\left[\widehat{u}\left(\frac{Y}{1-\lambda}\right)\right]-(1-\lambda)u(0)
\\
&\leq&
-\lambda u(x_0)+(1-\lambda)(1-u(0)) \\
&\leq&
-\alpha\,.
\end{eqnarray*}
As a result, $X+Y\in\cA_u$ whenever $\left\|Y\right\|_{\widehat{u}}<1-\lambda$,
showing that~$X$ belongs to the interior of~$\cA_u$.

\smallskip

To prove the ``only if'' part, assume $u(x)\leq\alpha$ for all $x\in\R$. Fix
$X\in\cA_u$ and $r>0$. We claim that $Y\notin\cA_u$ for some $Y\in H^{\widehat{u}}$
with $\left\|Y-X\right\|_{\widehat{u}}\leq r$. To this end, take $\gamma>0$ such that
$\probp(\left|X\right|\leq\gamma)>0$ and $\lambda>0$ for which
$u(\gamma-\lambda)<\alpha$. Since $(\Omega,\cF,\probp)$ is nonatomic, we can find a
measurable set $A\subset\{\left|X\right|\leq\gamma\}$ satisfying
$\widehat{u}\left(\frac{\lambda}{r}\right)\probp(A)\leq1$. Hence, setting
$Y:=(X-\lambda)1_A+X1_{A^c}$ it follows that $\left\|Y-X\right\|_{\widehat{u}}\leq
r$. Moreover, since $u(\gamma-\lambda)<\alpha$, we obtain
\begin{equation*}
\E[u(Y)] \leq u(\gamma-\lambda)\probp(A)+\alpha\,\probp(A^c)<\alpha\,.
\end{equation*}
showing that $Y\notin\cA_u$.

\smallskip

\textit{(ii)} Assume $u$ is bounded from above, and let $(X_n)$ be a sequence in
$\cA_u$ converging to $X$. Without loss of generality, we can assume $X_n\to X$
almost surely. Since $u$ is bounded from above, it follows from Fatou's Lemma 11.20
in \cite{AliprantisBorder2006} that $X\in\cA_u$.
\end{proof}

\medskip

We can now provide a complete characterization of finiteness and continuity for risk
measures based on shortfall risk on Orlicz hearts.

\medskip

\begin{proposition}
\label{risk measure u finite}
Consider a traded asset~$S=(S_0,S_T)$.
\begin{enumerate}[(i)]
	\item If~$u$ is bounded from above, the following statements are equivalent:
\begin{enumerate}[(a)]
	\item $\rho_{\cA_u,S}$ is finitely valued;
	\item $u(x_0)>\alpha$ for some $x_0>0$ and $\probp(S_T=0)=0$.
\end{enumerate}
  In this case, $\rho_{\cA_u,S}$ is continuous.

  \item If~$u$ is not bounded from above, then $\rho_{\cA_u,S}$ is always finitely
      valued and continuous.
\end{enumerate}
\end{proposition}
\begin{proof}
\textit{(i)} Assume~\textit{(a)} holds. Then~$\cA_u$ must have nonempty core by
Proposition~\ref{lemma on internal point and domain}, and hence nonempty interior by
Proposition~\ref{interiors of acceptance}. As a result, Lemma~\ref{charact nonempty
int in orlicz} implies $u(x_0)>\alpha$ for some $x_0>0$. Assume now that
$\probp(S_T=0)>0$. Since $u(-\infty)=-\infty$, taking $\xi>0$ large enough we obtain
for all $\lambda\in\R$
\begin{equation*}
\E[u(-\xi1_\Omega+\lambda S_T)] \leq u(-\xi)\,
\probp(S_T=0)+\sup_{x\in\R}u(x)\,\probp(S_T>0) < \alpha\,.
\end{equation*}
As a result $\rho_{\cA_u,S}(-\xi1_\Omega)=\infty$, contradicting~\textit{(a)}.
Hence,~\textit{(a)} implies~\textit{(b)}.

\smallskip

To prove the converse implication, assume~\textit{(b)} holds. Note
that~$\probp(S_T=0)=0$ implies that~$S_T$ is a strictly positive element
in~$H^{\widehat{u}}$. Moreover,~$\cA_u$ has nonempty interior by Lemma~\ref{charact
nonempty int in orlicz}. Hence,~\textit{(a)} follows immediately from
Corollary~\ref{finiteness for convex acc}.

\smallskip

\textit{(ii)} First, we show that~$\rho_{\cA_u,S}$ never attains~$-\infty$. Indeed,
assume $\rho_{\cA_u,S}(X)=-\infty$ for some $X\in H^{\widehat{u}}$, and take
$\beta<0$ such that
\begin{equation}
\label{auxiliary ineq utility}
\beta\,\probp(S_T>0)+\E[u(X)1_{\{u(X)>0\}}]<\alpha\,.
\end{equation}
Since $v:=u\vee\beta\ge u$, we have $\E[v(X-nS_T)]\geq\alpha$ for every positive
integer~$n$. Hence, using dominated convergence it is easy to show that
\begin{equation*}
\beta\,\probp(S_T>0)+\E[u(X)1_{\{u(X)>0\}}] \geq \E[\beta
1_{\{S_T>0\}}+v(X)1_{\{S_T=0\}}] \geq \alpha\,.
\end{equation*}
But this contradicts~\eqref{auxiliary ineq utility}, showing that~$\rho_{\cA_u,S}$
cannot attain the value~$-\infty$.

\smallskip

Now, since~$u$ is not bounded from above, we always have $u(x_0)>\alpha$ for some
$x_0>0$, hence the interior of~$\cA_u$ is nonempty by Lemma~\ref{charact nonempty int
in orlicz}. Finally, take $\gamma>0$ so that $\probp(S_T>\gamma)>0$. For any $\xi>0$
we can find $\lambda>0$ for which
\begin{equation*}
\E[u(-\xi1_\Omega+\lambda S_T)] \geq
u(-\xi)\,\probp(S_T\leq\gamma)+u(-\xi+\lambda\gamma)\,\probp(S_T>\gamma) \geq
\alpha\,,
\end{equation*}
showing that $\rho_{\cA_u,S}(-\xi1_\Omega)<\infty$. Since~$1_\Omega$ is a strictly
positive element, we can apply Theorem~\ref{general finiteness convex with str pos
elements} to find that~$\rho_{\cA_u,S}$ is finitely valued and continuous,
concluding the proof.
\end{proof}

\medskip

Note that Lemma~\ref{charact nonempty int in orlicz} continues to hold if the
underlying reference space is taken to be the Orlicz space~$L^{\widehat{u}}$. Our
next example shows that, when~$u$ is the exponential utility function, the risk
measure~$\rho_{\cA_u,S}$ is never finitely valued on $L^{\widehat{u}}$ even though
its domain has nonempty interior. This extends the cash-additive example by Biagini
and Frittelli at the end of Section~5.1 in~\cite{BiaginiFrittelli2009}, which was
used to highlight that the results in Cheridito and Li~\cite{CheriditoLi2009} for
Orlicz hearts are not valid in the context of general Orlicz spaces.

\medskip

\begin{example}
\label{no strictly positive}
Let $u(x):=1-e^{-x}$ be the exponential utility function. For any traded asset
$S=(S_0,S_T)$ the risk measure~$\rho_{\cA_u,S}$ is not finitely valued
on~$L^{\widehat{u}}$. Indeed, since $(\Omega,\cF,\probp)$ is nonatomic, we can always
find $Y\in L^{1/2}\setminus L^1$ such that $Y\geq1$ almost surely and~$Y$ is
independent of~$S_T$. Setting $X:=-\log(Y)$, it is easy to see that $X\in
L^{\widehat{u}}$ and $\E[e^{-X}]=\infty$. As a consequence, for any $\lambda\in\R$ we
have
\begin{equation*}
\E[u(X+\lambda S_T)] = 1-\E[e^{-X}]\,\E[e^{-\lambda S_T}] = -\infty < \alpha
\end{equation*}
showing that $\rho_{\cA_u,S}(X)=\infty$.
\end{example}

\medskip

\begin{remark}
\label{no strictly positive orlicz}
\begin{enumerate}[(i)]
\item The previous example has an interesting consequence. Since~$\cA_u$ is convex
    and has nonempty interior in~$L^{\widehat{u}}$, Corollary~\ref{finiteness for
    convex acc} implies that the Orlicz space~$L^{\widehat{u}}$ has no strictly
    positive elements. We do not know whether, more generally, whenever a
    nontrivial Orlicz heart and the corresponding Orlicz space do not coincide, the
    Orlicz space does not possess strictly positive elements.
\item Note that $1_\Omega$ is a strictly positive element in every nontrivial
    Orlicz heart but not in a general Orlicz space, unless the two coincide. This
    is the fundamental reason why the results in Cheridito and
    Li~\cite{CheriditoLi2009} on Orlicz hearts are not always applicable in the
    context of general Orlicz spaces.
\end{enumerate}
\end{remark}


\section{Cash subadditivity and quasiconvexity}
\label{cash sub and quasi conv}

This final section is devoted to discussing the link between the general risk
measures studied in this paper and cash-subadditive as well as quasiconvex risk
measures. Establishing this link is important since our framework provides a natural
setting to deal with a defaultable reference asset and cash subadditivity was
introduced in~\cite{ElKarouiRavanelli2009} to address the possible defaultability of
the given reference bond. Moreover, quasiconvexity arises naturally in the presence
of a convex acceptability criterion when the reference asset is not liquidly traded.


\subsection{Cash subadditivity and defaultable assets}
\label{cash-sub section}

In the recent influencial paper~\cite{ElKarouiRavanelli2009}, El Karoui and Ravanelli
questioned the axiom of cash additivity and introduced the new class of (convex)
cash-subadditive risk measures on~$L^\infty$ in order to ``model stochastic and/or
ambiguous interest rates or defaultable contingent claims''. In a more general
setting, i.e. working in $L^p$, $0\le p\le\infty$, and without requiring finiteness
or convexity, a \textit{cash-subadditive} risk measure is defined as a decreasing
function $\rho:L^p\to\overline{\R}$ satisfying
\begin{equation*}
\rho(X+\lambda1_\Omega)\geq\rho(X)-\lambda \ \ \ \ \mbox{for all} \ \lambda>0 \
\mbox{and} \ X\in L^p\,.
\end{equation*}

\medskip

Consider an acceptance set $\cA\subset L^p$ and a traded asset~$S=(S_0,S_T)$. Note that in the
framework considered thus far the $S$-additivity of~$\rho_{\cA,S}$ is a direct
consequence of the fact that the price of~$\lambda$ units of~$S$ is $\lambda S_0$.
Hence, unless we assume a nonlinear pricing rule as we will do in the final section
of this paper,~$\rho_{\cA,S}$ will always be $S$-additive. Consequently, cash
subadditivity is not a surrogate for $S$-additivity but rather a
property~$\rho_{\cA,S}$ may or may not have. If we wish to interpret risk measures as
capital requirements which measure the distance of future financial positions to
acceptability, any new property stipulated for risk measures, such as cash
subadditivity, needs to be justified by a corresponding financially meaningful
property of either~$\cA$ or~$S$. Therefore, in this section we investigate what
makes~$\rho_{\cA,S}$ be cash subadditive. By doing so, we also provide a better
financial insight into the axiom of cash subadditivity. In particular, our results
show that this assumption is typically not satisfied when the asset~$S$ is
defaultable, thus raising questions at least about part of the interpretation given
in~\cite{ElKarouiRavanelli2009}.

\medskip

Note that, by $S$-additivity, cash subadditivity of~$\rho_{\cA,S}$ is equivalent to
\begin{equation}
\label{equivalent cash subadditivity}
\rho_{\cA,S}(X+\lambda1_\Omega)\geq\rho_{\cA,S}\left(X+\frac{\lambda}{S_0}S_T\right)
\ \ \ \ \mbox{for all} \ \lambda>0 \ \mbox{and} \ X\in L^p\,.
\end{equation}

\medskip

Assume $S=(S_0,S_T)$ is a defaultable bond with face value~$1$, recovery rate $0\le
S_T\le 1$, and price $S_0<1$. Then we can interpret~$S_0$ as the \textit{invested
capital} and $1-S_0$ as the \textit{interest payment}. The following result shows
that if the invested capital is not at risk, i.e. if the bond can only default on the
interest payment, then the risk measure~$\rho_{\cA,S}$ is always cash subadditive.

\medskip

\begin{proposition}
\label{ex of cash-sub}
Let~$\cA\subset L^p$ be an acceptance set and $S=(S_0,S_T)$ a traded asset. Assume
$\probp(S_T<S_0)=0$. Then~$\rho_{\cA,S}$ is cash subadditive.
\end{proposition}
\begin{proof}
Taking $X\in L^p$ and $\lambda>0$ and noting that
$\lambda1_\Omega\leq\frac{\lambda}{S_0}S_T$, we immediately obtain
by~\eqref{equivalent cash subadditivity} and monotonicity that~$\rho_{\cA,S}$ is cash
subadditive.
\end{proof}

\medskip

We investigate now the case where the capital invested in the asset $S=(S_0,S_T)$ is
at risk, i.e. when $\probp(S_T<S_0)>0$. In this situation, we will see that cash
subadditivity is typically not satisfied. Moreover, cash subadditivity turns out to
depend not only on the payoff~$S_T$ of the eligible asset but also on the prevailing
price~$S_0$. As such, this property is not stable with respect to changes in the
price of the eligible asset, a circumstance which would seem to limit its practical
usefulness.

\medskip

We start by providing a necessary condition for cash subadditivity for a general
underlying acceptance set and a sufficient condition in the coherent case.

\medskip

\begin{proposition}
\label{charact cashsub coherent}
Let $\cA\subset L^p$ be an acceptance set containing $0$ and $S=(S_0,S_T)$ a traded
asset. The following statements hold:
\begin{enumerate}[(i)]
	\item if $\rho_{\cA,S}$ is cash subadditive, then
$S_T-S_01_\Omega\in\overline{\cA}$;
	\item if~$\cA$ is coherent and $S_T-S_01_\Omega\in\cA$, then $\rho_{\cA,S}$ is
cash subadditive.
\end{enumerate}
In particular, if $\cA$ is closed and coherent, then $\rho_{\cA,S}$ is cash
subadditive if and only if $S_T-S_01_\Omega\in\cA$.
\end{proposition}
\begin{proof}
To prove \textit{(i)}, assume cash subadditivity and note that taking
$X:=-S_01_\Omega$ and $\lambda:=S_0$ in \eqref{equivalent cash subadditivity} we
obtain
\begin{equation*}
\rho_{\cA,S}(S_T-S_01_\Omega) = \rho_{\cA,S}\left(X+\frac{\lambda}{S_0}S_T\right)
\leq \rho_{\cA,S}(X+\lambda1_\Omega) = \rho_{\cA,S}(0) \leq  0\,.
\end{equation*}
Hence, $S_T-S_01_\Omega\in\overline{\cA}$.

\smallskip

To prove~\textit{(ii)}, assume~$\rho_{\cA,S}$ is not cash subadditive. Then we find
$X\in L^p$ and $\lambda>0$ for which
\begin{equation*}
\rho_{\cA,S}(X+\lambda1_\Omega) < 0 <
\rho_{\cA,S}\left(X+\frac{\lambda}{S_0}S_T\right)
\end{equation*}
so that $X+\lambda1_\Omega\in\cA$ while $X+\frac{\lambda}{S_0}S_T\notin\cA$.
Since~$\cA$ is coherent and
\begin{equation*}
X+\frac{\lambda}{S_0}S_T = X+\lambda1_\Omega+\frac{\lambda}{S_0}(S_T-S_01_\Omega)
\end{equation*}
we immediately conclude that $S_T-S_01_\Omega\notin\cA$. This shows that~\textit{(ii)}
holds.
\end{proof}

\medskip

The following corollaries provide a characterization of cash subadditivity for risk
measures based on $\TVaR$-acceptability and scenario-based acceptability,
respectively. In particular, for acceptability based on Tail Value-at-Risk at level~$\alpha$ it turns out
that the corresponding risk measure~$\rho_{\cA^\alpha,S}$ is not cash subadditive
whenever the probability that the invested capital is at risk exceeds the
level~$\alpha$.

\medskip

\begin{corollary}[Tail Value-at-Risk]
Let $\cA^\alpha\subset L^p$ be the acceptance set based on Tail Value-at-Risk at
level $\alpha\in(0,1)$, and let $S=(S_0,S_T)$ be a traded asset. Then
$\rho_{\cA^\alpha,S}$ is cash subadditive if and only if
$\TVaR_\alpha(S_T)\leq-S_0$.

\smallskip

In particular, if $\rho_{\cA^\alpha,S}$ is cash subadditive then
$\probp(S_T<S_0)\leq\alpha$.
\end{corollary}

\medskip

\begin{corollary}[Scenarios]
Let $A\in\cF$ and define $\cA(A):=\{X\in L^p \,; \ X1_A\ge 0\}$. If $S=(S_0,S_T)$ is
a traded asset, then $\rho_{\cA(A),S}$ is cash subadditive if and only if
$\probp(A\cap\{S_T<S_0\})=0$.
\end{corollary}

\medskip

The following result shows that $\VaR$-based risk measures are never cash subadditive
as soon as the invested capital is at risk.

\medskip

\begin{proposition}[Value-at-Risk]
\label{no cash-sub for var}
Let $\cA_\alpha\subset L^p$ be the acceptance set based on Value-at-Risk at level
$\alpha\in(0,1)$ and consider a traded asset $S=(S_0,S_T)$. Then $\rho_{\cA_\alpha,S}$ is
cash subadditive if and only if~$\probp(S_T<S_0)=0$.
\end{proposition}
\begin{proof}
The ``if'' part follows from Proposition \ref{ex of cash-sub}. To prove the ``only
if'' part assume $\rho_{\cA_\alpha,S}$ is cash subadditive but $\probp(S_T<S_0)>0$.
Take $\e\in(0,1)$ such that $\probp(S_T\leq\e S_0)>0$. Since
$\probp(S_T<S_0)\leq\alpha$ by part \textit{(i)} in Proposition~\ref{charact cashsub
coherent}, we can find $0<\delta<\probp(S_T\leq\e S_0)$ satisfying
\begin{equation*}
\probp(S_T\geq
S_0) > \alpha+\delta-\probp(S_T<S_0) > 0\,.
\end{equation*}
Moreover, since $(\Omega,\cF,\probp)$ is nonatomic we have $\probp(A)=\alpha+\delta-\probp(S_T<S_0)$ for some measurable subset $A$ of $\{S_T\geq S_0\}$. Take $\gamma>0$ such that $(1+
\gamma)\e<1$ and set
\begin{equation*}
X:=\begin{cases}
-1 & \mbox{on $B$}\\
-\frac{2+\gamma}{S_0}S_T & \mbox{on $B^c$}
\end{cases}
\end{equation*}
where $B:=\{S_T\leq\e S_0\}\cup(\{S_T\ge S_0\}\setminus A)$. Then
$\probp(X+1_\Omega<0)\leq\probp(B^c)<\alpha$, so that
$\rho_{\cA_\alpha,S}(X+1_\Omega)\le 0$. Moreover,
\begin{equation*}
X+\frac{1}{S_0}S_T+\frac{\gamma}{S_0}S_T \leq -1+(1+\gamma)\e < 0 \ \ \ \mbox{on} \
\{S_T\le\e S_0\}
\end{equation*}
and
\begin{equation*}
X+\frac{1}{S_0}S_T+\frac{\gamma}{S_0}S_T = -\frac{1}{S_0}S_T < 0 \ \ \ \mbox{on} \
B^c\,.
\end{equation*}
Hence, it follows
\begin{equation*}
\probp\left(X+\frac{1}{S_0}S_T+\frac{\gamma}{S_0}S_T<0\right) \geq
\probp(S_T<S_0)+\probp(A) = \alpha+\delta > \alpha
\end{equation*}
showing that $\rho_{\cA_\alpha,S}(X+\frac{1}{S_0}S_T)\geq\gamma>0$. Since we had
already proved that $\rho_{\cA_\alpha,S}(X+1_\Omega)\leq0$, we conclude that
$\rho_{\cA_\alpha,S}$ cannot be cash subadditive by \eqref{equivalent cash
subadditivity}, contradicting our initial assumption.
\end{proof}

\medskip

\begin{remark}
The above Proposition~\ref{no cash-sub for var} shows that Proposition~\ref{charact
cashsub coherent} fails if we drop the assumption of convexity. In fact, the
assumption of conicity cannot be dropped either. Indeed, consider the acceptance set $\cA:=\{X\in L^p
\,; \ \E[X]\geq\alpha\}$ for a fixed $\alpha\in\R$, and let $S=(S_0,S_T)$ be a traded
asset. Then it is easy to see that $\rho_{\cA,S}$ is cash subadditive if and only if
$\E[S_T]\geq S_0$.
\end{remark}


\subsection{Quasiconvexity and illiquid eligible assets}
\label{illiquidity section}

We now proceed to extend in a natural way the definition of a risk measure to account
for situations where the eligible asset is not liquidly traded. We assume~$\cX$ is a
general ordered topological vector space.

\medskip

For a liquidly traded asset $S=(S_0,S_T)$, the price of $\lambda\in\R$ units of~$S$
is $\pi(\lambda):=\lambda S_0$, where~$S_0$ is the unit price. When the asset is not
liquidly traded the \textit{pricing functional} $\pi:\R\to\R$ will no longer be
linear. Here, we only assume~$\pi$ is strictly increasing and satisfies $\pi(0)=0$
and $\pi(1)=S_0$. Note that we do not require any form of continuity for $\pi$ thus
allowing for price jumps, which are a typical feature of an illiquid market.

\medskip

\begin{definition}
Let $\cA\subset\cX$ be an acceptance set and $S=(S_0,S_T)$ a traded asset with
pricing functional~$\pi$. The \textit{risk measure with respect to}~$\cA$,~$S$
{\em and}~$\pi$ is the function $\rho_{\cA,S,\pi}:\cX\to\overline{\R}$ defined by
\begin{equation*}
\rho_{\cA,S,\pi}(X):=\inf\{\pi(\lambda) \,; \ \lambda\in\R \,:\, X+\lambda
S_T\in\cA\}\,.
\end{equation*}
\end{definition}

\medskip

From now on we assume that $\cA$ is closed. In this case we can reduce risk measures
with respect to illiquid eligible assets to risk measures of the form~$\rho_{\cA,S}$.
Note that, even if the asset~$S$ is not liquidly traded, we will still
write~$\rho_{\cA,S}$ to denote the risk measure we would get if we assumed full
liquidity. As usual we set $\pi(\pm\infty):=\lim_{\lambda\to\pm\infty}\pi(\lambda)$.

\medskip

\begin{lemma}
\label{lemma illiquid}
Let $\cA\subset\cX$ be a closed acceptance set and $S=(S_0,S_T)$ a traded asset with
pricing functional~$\pi$. Then for all $X\in\cX$
\begin{equation}
\label{switch pi and rho}
\rho_{\cA,S,\pi}(X)=\pi\left(\frac{1}{S_0}\rho_{\cA,S}(X)\right)\,.
\end{equation}
\end{lemma}
\begin{proof}
Take $X\in\cX$. Since~$\cA$ is closed, the set of all $\lambda\in\R$ satisfying
$X+\lambda S_T\in\cA$ is a closed interval, possibly the full real line. As a result
of the monotonicity of~$\pi$, the equality~\eqref{switch pi and rho} follows.
\end{proof}

\medskip

\begin{remark}
\label{remark on lsc of illiquid measures}
\begin{enumerate}[(i)]
\item Note that, since $\cA$ is closed,~$\rho_{\cA,S}$ is lower semicontinuous by
    Remark~\ref{KainaRuschendorf}. However, this need not be the case for
    $\rho_{\cA,S,\pi}$ unless $\pi$ is left continuous.

\item The above proposition should be compared with Example~2.2
    in~\cite{Cerreia2011} where the payoff of the reference asset is
    $S_T=1_\Omega$. There, formula~\eqref{switch pi and rho} is obtained by
    requiring the upper semicontinuity of~$\pi$ rather than the closedeness
    of~$\cA$.
\end{enumerate}
\end{remark}

\medskip

The property of quasiconvexity is known to be the minimal property a risk measure
needs to have to capture diversification effects. Quasiconvexity of risk measures has
been extensively studied for instance in Cerreia-Vioglio, Maccheroni, Marinacci and
Montrucchio~\cite{Cerreia2011} and Drapeau and Kupper~\cite{DrapeauKupper2013}.
Recall that a function $\rho:\cX\to\overline{\R}$ is called \textit{quasiconvex}
whenever the level set $\{X\in\cX \,; \ \rho(X)\leq\alpha\}$ is convex for every
$\alpha\in\R$. As for the cash-additive case, it is easy to see that for an
$S$-additive risk measure quasiconvexity is equivalent to convexity. Hence, genuine
quasiconvexity can only be observed if the pricing rule for $S$ is not linear. The
next proposition provides a characterization of quasiconvex risk measures of the form
$\rho_{\cA,S,\pi}$.

\medskip

\begin{proposition}
\label{prop on qc and illiquid}
Let $\cA\subset\cX$ be a closed acceptance set and $S=(S_0,S_T)$ a traded asset with
pricing functional~$\pi$. Then~$\rho_{\cA,S,\pi}$ is quasiconvex if and only if~$\cA$
is convex.
\end{proposition}
\begin{proof}
Assume $\cA$ is convex. Being the composition of a convex and an increasing function by~\eqref{switch pi and rho}, the map~$\rho_{\cA,S,\pi}$ is quasiconvex. On the other hand, if $\rho_{\cA,S,\pi}$ is quasiconvex, then $\cB:= \{X\in\cX \,; \ \rho_{\cA,S,\pi}(X)\le0\}$ is convex. Clearly, $\cA\subset\cB$. Take $X\in\cB$ and note that $\pi(\frac{1}{S_0}\rho_{\cA,S}(X))=\rho_{\cA,S,\pi}(X)\le 0$. Since $\pi(0)=0$ and
$\pi$ is strictly increasing we immediately obtain $\rho_{\cA,S}(X)\le 0$. But this
implies that $X\in\cA$ as a consequence of Lemma~\ref{translation and
monotonicity}.
\end{proof}

\medskip

We now provide a dual representation for quasiconvex risk measures of the
form~$\rho_{\cA,S,\pi}$. In contrast to Proposition~5 in Drapeau and
Kupper~\cite{DrapeauKupper2013}, we do not require lower semicontinuity but exploit
the special structure of~$\rho_{\cA,S,\pi}$. In particular, our representation
formula~\eqref{quasiconvex representation} below allows for a transparent
interpretation in terms of the fundamental financial primitives: the acceptance set,
the eligible asset, and the pricing functional. Recall that the \textit{Fenchel
conjugate} of a map $\rho:\cX\to\overline{\R}$ is the function
$\rho^\ast:\cX'\to\overline{\R}$ defined by
\begin{equation*}
\rho^\ast(\psi):=\sup_{X\in\cX}\{\psi(X)-\rho(X)\}\,.
\end{equation*}
Moreover, for a traded asset $S=(S_0,S_T)$ we introduce the set
\begin{equation*}
\cX'_{+,S}:=\{\psi\in\cX'_+ \,; \ \psi(S_T)=S_0\}\,.
\end{equation*}

\medskip

\begin{proposition}
\label{corollary on dual illiquid}
Assume~$\cX$ is locally convex. Let $\cA\subset\cX$ be a closed, convex acceptance
set and $S=(S_0,S_T)$ a traded asset with pricing functional~$\pi$. If~$\rho_{\cA,S}$
is continuous at an interior point~$X$ of its domain of finiteness, then
\begin{equation}
\label{quasiconvex representation}
\rho_{\cA,S,\pi}(X)=\max_{\psi\in\cX'_{+,S}}\left\{\pi\left(\frac{\psi(-X)-\rho^\ast_{\cA,S}(-\psi)}{S_0}\right)\right\}\,.
\end{equation}
\end{proposition}
\begin{proof}
Being finite at~$X$, the convex and lower semicontinuous map~$\rho_{\cA,S}$ cannot
attain the value~$-\infty$ by Proposition~2.4 in~\cite{EkelandTemam1999}. Then,
following the lines of the proof of Corollary~7 in~\cite{FrittelliGianin2002}, we
obtain the standard dual representation
\begin{equation}
\label{dual representation formula}
\rho_{\cA,S}(X)=\sup_{\psi\in\cX'_{+,S}}\{\psi(-X)-\rho_{\cA,S}^\ast(-\psi)\}\,.
\end{equation}
Now take $m>\rho_{\cA,S}(X)$. Since~$\rho_{\cA,S}$ is continuous at~$X$, the interior
of~$\cA$ is nonempty and $X_m:=X+\frac{m}{S_0}S_T\in\Interior(\cA)$ by
Lemma~\ref{translation and monotonicity}. Note that
$\widetilde{X}:=X+\frac{\rho_{\cA,S}(X)}{S_0}S_T\in\partial\cA$. As a result,
Lemma~7.7 in~\cite{AliprantisBorder2006} implies that~$\widetilde{X}$ is a support
point of~$\cA$. Hence,
\begin{equation*}
\varphi(\widetilde{X})=\inf_{Z\in\cA}\varphi(Z)
\end{equation*}
for some nonzero functional $\varphi\in\cX'$, which is positive by
Lemma~\ref{halfspaces containing acceptance sets}. Moreover, $\varphi(S_T)>0$ since otherwise $\varphi(X_m)=\inf_{Y\in\cA}\varphi(Y)$ which is not
possible because $X_m\in\Interior(\cA)$. Rescaling~$\varphi$ we may assume that
$\varphi(S_T)=S_0$. Finally, taking $Y\in\Dom(\rho_{\cA,S})$ we can conclude that
\begin{equation*}
\varphi(X)+\rho_{\cA,S}(X) = \varphi(\widetilde{X}) = \inf_{Z\in\cA}\varphi(Z) \leq
\varphi\left(Y+\frac{\rho_{\cA,S}(Y)}{S_0}S_T\right) = \varphi(Y)+\rho_{\cA,S}(Y)\,.
\end{equation*}
Hence, the supremum in~\eqref{dual representation formula} is attained at~$\varphi$
so that~\eqref{quasiconvex representation} now easily follows from Lemma~\ref{lemma
illiquid} and from the monotonicity of~$\pi$.
\end{proof}

\medskip

\begin{remark}
Note that our attainability result is not implied by Theorem~1
in~\cite{BiaginiFrittelli2009} since we do not assume~$\cX$ is a Frech\'{e}t lattice.
In particular, it remains valid on $L^p$ spaces when equipped with any weak topology,
for instance on $L^\infty$ endowed with the weak topology $\sigma(L^\infty,L^1)$.
\end{remark}

\medskip

We now assume $\cX=L^p$ over a fixed probability space $(\Omega,\cF,\probp)$.
In~\cite{Cerreia2011} it is suggested that cash subadditivity may arise when the
reference (risk-free) asset is illiquidly traded. Since the most natural way to
incorporate illiquidity is by considering nonlinear pricing functionals as above, it
is interesting to see whether our risk measures $\rho_{\cA,S,\pi}$ are cash
subadditive. The following result shows that, in the common situations, this can only
be true if the pricing functional~$\pi$ is continuous, thus ruling out examples of
illiquid markets where the pricing rule may have jumps.

\medskip

\begin{proposition}
\label{continuity of pricing functional}
Fix $1\le p\le\infty$ and let~$\cA\subset L^p$ be a closed, convex acceptance set
with nonempty interior. Consider a traded asset $S=(S_0,S_T)$ with pricing functional
$\pi$. If~$\rho_{\cA,S,\pi}$ is cash subadditive and $\rho_{\cA,S}(0)\in\R$,
then~$\pi$ is continuous.
\end{proposition}
\begin{proof}
We first prove that for every $\lambda>0$
\begin{equation}
\label{continuity hilfe}
-\infty< \rho_{\cA,S}(\lambda 1_\Omega)<\rho_{\cA,S}(0)\,.
\end{equation}
Fix $\lambda>0$ and note that, being convex and lower semicontinuous, the risk
measure~$\rho_{\cA,S}$ cannot assume the value~$-\infty$ by Proposition~2.4
in~\cite{EkelandTemam1999}. Moreover, $X:=\frac{\rho_{\cA,S}(0)}{S_0} S_T$ belongs
to~$\partial\cA$ so that $X+\lambda 1_\Omega\in\Interior(\cA)$. Indeed, assume to the
contrary that $X+\lambda 1_\Omega\in\partial\cA$. Since~$X$ is a support point
of~$\cA$ by Lemma~7.7 in~\cite{AliprantisBorder2006}, we find a nonzero $\psi\in\cX'$
with $\psi(X+\lambda 1_\Omega)\leq\psi(Y)$ for all $Y\in\cA$. In particular, choosing
$Y:=X$ we get $\psi(1_\Omega)\le 0$ which is impossible because $1_\Omega$ is a
strictly positive element in $L^p$ and $\psi$ must be positive by
Lemma~\ref{halfspaces containing acceptance sets}. In conclusion, we obtain
\begin{equation*}
\rho_{\cA,S}(\lambda 1_\Omega)-\rho_{\cA,S}(0) = \rho_{\cA,S}(X+\lambda 1_\Omega) <
0\,,
\end{equation*}
proving~\eqref{continuity hilfe}.

\smallskip

Now, assume that~$\pi$ is not left continuous at $x_0\in\R$ and let
$\gamma:=\pi(x_0)-\lim_{x\uparrow x_0}\pi(x)>0$ be the size of the jump. Take
$\xi\in\R$ such that, setting $X:=\xi S_T$, we have
$\rho_{\cA,S}(X)=\rho_{\cA,S}(0)-\xi S_0=x_0 S_0$. Hence, $\rho_{\cA,S}(X+\lambda
1_\Omega)=\rho_{\cA,S}(\lambda 1_\Omega)-\xi S_0<x_0 S_0$ for any $\lambda>0$
by~\eqref{continuity hilfe}. But then for $\lambda\in(0,\gamma)$
\begin{equation*}
\rho_{\cA,S,\pi}(X)-\rho_{\cA,S,\pi}(X+\lambda 1_\Omega) =
\pi(x_0)-\pi\left(\frac{\rho_{\cA,S}(X+\lambda 1_\Omega)}{S_0}\right) \geq \gamma > \lambda
\end{equation*}
showing that $\rho_{\cA,S,\pi}$ is not cash subadditive.

\smallskip

Similarly, assume~$\pi$ is not right continuous at $x_0\in\R$ and set
$\gamma:=\lim_{x\downarrow x_0}\pi(x)-\pi(x_0)>0$. For $\lambda\in(0,\gamma)$ we find
$\xi\in\R$ such that, setting $X:=\xi S_T$, we have $\rho_{\cA,S}(X+\lambda
1_\Omega)=\rho_{\cA,S}(\lambda 1_\Omega)-\xi S_0=x_0 S_0$. Since
$\rho_{\cA,S}(X)=\rho_{\cA,S}(0)-\xi S_0>x_0 S_0$ by~\eqref{continuity hilfe}, we
conclude that
\begin{equation*}
\rho_{\cA,S,\pi}(X)-\rho_{\cA,S,\pi}(X+\lambda 1_\Omega) =
\pi\left(\frac{\rho_{\cA,S}(X)}{S_0}\right)-\pi(x_0) \geq \gamma > \lambda\,.
\end{equation*}
Again, this implies that $\rho_{\cA,S,\pi}$ is not cash subadditive. In conclusion,
if~$\rho_{\cA,S,\pi}$ has to be cash subadditive, the pricing functional must be
continuous.
\end{proof}


\end{document}